\documentclass[sigconf, nonacm]{acmart}

\newcommand\vldbdoi{XX.XX/XXX.XX}
\newcommand\vldbpages{XXX-XXX}
\newcommand\vldbvolume{14}
\newcommand\vldbissue{1}
\newcommand\vldbyear{2021}
\newcommand\vldbauthors{\authors}
\newcommand\vldbtitle{\shorttitle} 
\newcommand\vldbavailabilityurl{URL_TO_YOUR_ARTIFACTS}
\newcommand\vldbpagestyle{plain} 

\usepackage{subcaption} 
\usepackage{autobreak}
\usepackage{multirow}
\usepackage{url}
\usepackage{array}
\AtBeginDocument{%
  \providecommand\BibTeX{{%
    \normalfont B\kern-0.5em{\scshape i\kern-0.25em b}\kern-0.8em\TeX}}}

\newcounter{example}
\newenvironment{example}%
   {
     \par\refstepcounter{example}\textbf{EXAMPLE \theexample}\\}%
   {\par}

\begin{document}

\title{Coo: Rethink Data Anomalies In Databases}



\author{Haixiang Li$^\ast$, Xiaoyan Li$^\dagger$, Yuxing Chen$^\ast$, Xiaoyong Du$^\ddagger$, Wei Lu$^\ddagger$, Chang Liu$^\ast$, Yuean Zhu$^\ast$, Anqun Pan$^\ast$}
\affiliation{%
  \institution{$\ast$ Tencent Inc. $\dagger$ Peking University $\ddagger$ Renmin University of China}
}
\email{{blueseali,axingguchen,williamcliu,anduinzhu,aaronpan}@tencent.com;li\_xiaoyan@pku.edu.cn;  {duyong,lu-wei}@ruc.edu.cn}

\renewcommand{\shortauthors}{Trovato and Tobin, et al.}

\begin{abstract}
Transaction processing technology has three important contents: data anomalies, isolation levels, and concurrent control algorithms. Concurrent control algorithms are used to eliminate some or all data anomalies at different isolation levels to ensure data consistency.
Isolation levels in the current ANSI/ISO SQL standard are defined by disallowing certain kinds of data anomalies. However, the definitions of data anomalies in the ANSI/ISO SQL standard are controversial.
On the one hand, the definitions lack a mathematical formalization and cause ambiguous interpretations.
On the other hand, the definitions are made in a case-by-case manner, and lead to a situation that even a senior DBA could not have infallible knowledge of data anomalies, due to a lack of a complete understanding of its nature.
While revised definitions in existing literature propose various mathematical formalizations to correct the former argument, how to address the latter argument remains an open problem.
In this paper, we present a general framework called {\bf Coo} with the capability to formally define all data anomalies and quantitatively describe.
Under this framework, we show that existing reported data anomalies are only a small portion.
While we theoretically prove that Coo is complete to formalize data anomalies mathematically, we employ a novel method to classify infinite data anomalies. In addition, we use this framework to define new isolation levels and quantitatively describe the proportion of data anomalies and the rollback rates of mainstream concurrency control algorithms..
These works show that the C and I of ACID can be quantitatively analyzed based on all data anomalies.


\end{abstract}





\maketitle

\pagestyle{\vldbpagestyle}
\begingroup\small\noindent\raggedright\textbf{PVLDB Reference Format:}\\
\vldbauthors. \vldbtitle. PVLDB, \vldbvolume(\vldbissue): \vldbpages, \vldbyear.\\
\href{https://doi.org/\vldbdoi}{doi:\vldbdoi}
\endgroup
\begingroup
\renewcommand\thefootnote{}\footnote{\noindent
This work is licensed under the Creative Commons BY-NC-ND 4.0 International License. Visit \url{https://creativecommons.org/licenses/by-nc-nd/4.0/} to view a copy of this license. For any use beyond those covered by this license, obtain permission by emailing \href{mailto:info@vldb.org}{info@vldb.org}. Copyright is held by the owner/author(s). Publication rights licensed to the VLDB Endowment. \\
\raggedright Proceedings of the VLDB Endowment, Vol. \vldbvolume, No. \vldbissue\ %
ISSN 2150-8097. \\
\href{https://doi.org/\vldbdoi}{doi:\vldbdoi} \\
}\addtocounter{footnote}{-1}\endgroup


\section{Introduction}
The data anomalies are not recognized to the due height. 
In this section, we discuss why we should define all data anomalies.

The ANSI/ISO SQL \cite{16} specifies four data anomalies, including \textit{Dirty Write}, \textit{Dirty Read}, \textit{Non-repeatable Read}, and \textit{Phantom}.
Under the specifications of data anomalies, four isolation levels, i.e., \textit{Read Uncommitted}, \textit{Read Committed}, \textit{Repeatable Read} and \textit{Serializable}, are defined accordingly to eliminate the different anomalies.
For example, the Read Committed isolation level is defined by disallowing Dirty Write and Dirty Read.

\begin{table*}[]
\caption{A Thorough Survey on Data Anomalies in Existing Literature}
\footnotesize
\begin{tabular}{l|c|l}
\toprule
\hline
\textbf{No} &
  \textbf{Anomaly, reference, year} &
  \textbf{Formal definition}  \\ \hline

1   &  \begin{tabular}[c]{@{}c@{}}Dirty Write \cite{16} 1992  \end{tabular} &
  \begin{tabular}[c]{@{}l@{}}$W_1[x_1]...W_2[x_2]...$(($C_1$ or $ A_1$) and  ($C_2$ or $A_2$) in any order)\end{tabular} \\ \hline

2   &   \begin{tabular}[c]{@{}c@{}}Lost Update \cite{10.1145/223784.223785} 1995   \end{tabular} &
   $R_1[x_0]$...$W_2[x_1]$...$C_2$...$W_1[x_2]$ \\\hline

3   &   \begin{tabular}[c]{@{}c@{}}Dirty Read \cite{16} 1992  \end{tabular} &
   $W_1[x]...R_2[x]$...($A_1$ and $C_2$ in either order) \\\hline

4   &   \begin{tabular}[c]{@{}c@{}}Aborted Reads \cite{xie2015high} 2015, \cite{839388} 2000 \end{tabular} &   $W_1[x:i]...R_2[x:i]$...($A_1$ and $C_2$ in any order) \\\hline

5   &   \begin{tabular}[c]{@{}c@{}}Fuzzy/Non-repeatable Read \cite{16} 1992    \end{tabular} &   $R_1[x]...W_2[x]...C_2...R_1[x]...C_1$  \\\hline

6   &    \begin{tabular}[c]{@{}c@{}}Phantom \cite{16} 1992 \end{tabular} &   $R_1[P]$...$W_2${[}$y$ in $P${]}...$C_2$...$R_1[P]$...$C_1$ \\\hline

7   &   \begin{tabular}[c]{@{}c@{}}Intermediate Reads \cite{xie2015high} 2015, \cite{839388} 2000   \end{tabular} &   $W_1[x:i]...R_2[x:i]...W_1[x:j]$...$C_2$ \\\hline

8   &   \begin{tabular}[c]{@{}c@{}}Read Skew \cite{10.1145/223784.223785} 1995  \end{tabular} &   $R_1[x_0]...W_2[x_1]...W_2[y_1]...C_2...R_1[y_1]$ \\\hline

9   &   \begin{tabular}[c]{@{}c@{}}Unnamed Anomaly \cite{schenkel2000federated} 2000   \end{tabular} &   \begin{tabular}[c]{@{}l@{}}$R_1[y]...R_2[x]...W_2[x]...R_2[y]...W_2[y]...C_2...R_3[x]...W_3[x]...R_3[z]...W_3[z]...C_3...R_1[z]...C_1$\end{tabular} \\\hline

10   &   \begin{tabular}[c]{@{}c@{}}Fractured Reads \cite{cerone2017algebraic} 2017, \cite{10.1145/2909870} 2014  \end{tabular} &
   $ R_1[x_0]...W_2[x_1]...W_2[y_1]...C_2...R_1[y_1] $ \\\hline

11   &   \begin{tabular}[c]{@{}c@{}}Serial-concurrent-phenomenon \cite{binnig2014distributed}   2014   \end{tabular} &   $R_1[x_0]...W_2[x_1]...W_2[y_1]...C_2...R_1[y_1] $ \\\hline

12   &   \begin{tabular}[c]{@{}c@{}}Cross-phenomenon \cite{binnig2014distributed} 2014  \end{tabular} &   \begin{tabular}[c]{@{}l@{}}$R_1[x_0]...R_2[y_0]...W_3[x_1]...C_3... W_4[y_1]... C_4...R_2[x_1]...R_1[y_1]$\end{tabular} \\\hline

13   &   \begin{tabular}[c]{@{}c@{}}Long Fork Anomaly \cite{cerone2017algebraic} 2017   \end{tabular} &   \begin{tabular}[c]{@{}l@{}}$R_1[x_0]...R_2[y_0]...W_3[x_1]...C_3... W_4[y_1]... C_4...R_2[x_1]...R_1[y_1]$\end{tabular} \\\hline

14   &   \begin{tabular}[c]{@{}c@{}}Causality Violation Anomaly \cite{cerone2017algebraic} 2017    \end{tabular} &
   \begin{tabular}[c]{@{}l@{}}$R_1[x_0]...W_2[x_1]...C_2...R_3[x_1]... W_3[y_1]...C_3...R_1[y_1]$\end{tabular} \\\hline

15   &   \begin{tabular}[c]{@{}c@{}}Read-only Transaction Anomaly \cite{10.1145/1031570.1031573, Read_Only_Transactions} 2004   \end{tabular} &
   \begin{tabular}[c]{@{}l@{}}$R_1[x_0,0]...R_1[y_0,0]...R_2[y_0,0]... W_2[y_1,20]...  C_2...R_3[x_0,0]...R_3[y_1,20]...C_3...W_1[x_2,-11]... C_1$\end{tabular} \\\hline

16   &   \begin{tabular}[c]{@{}c@{}}Write Skew \cite{10.1145/223784.223785} 1995   \end{tabular} &   $R_1[x_0]...R_2[y_0]...W_1[y_1]...W_2[x_1]$ \\\hline

17   &   \begin{tabular}[c]{@{}c@{}}Predicate-based Write Skew \cite{DBLP:journals/tods/FeketeLOOS05} 2005    \end{tabular} & $R_1[P]...R_2[P]...W_1$ [$y_1$ in $P$]$...W_2$[$x_1$ in $P$] \\\hline

18   &  \begin{tabular}[c]{@{}c@{}}Read Partial-committed \cite{duxiaoyong_read_partial_committed} 2019    \end{tabular} &  $R_1 [x]\dot W_2 [x]\dot W_2 [y]\dot C_2 \dot R_1 [y]\dot C_1$\\\hline

\bottomrule
\end{tabular}

\label{tab:datano}
\end{table*}

However, the definitions of data anomalies in the ANSI/ISO SQL standard are controversial.
On the one hand, the definitions lack a mathematical formalization, and cause ambiguous interpretations.
Let symbol “$W_1(x, k)$” be a write by transaction $T_1$ on data item $x$ with value $k$, and “$R_2(x, k)$” be a read on $x$ returned with value $k$ by transaction $T_2$.
Transaction $T_1$'s commit or abort are written as “$C_1$” or “$A_1$”, respectively.
For example in Dirty Read,
it is specified as: $T_1$ writes $x$; $T_2$ then reads $x$ before $T_1$ performs a commit or abort; if $T_1$ then performs a abort, $T_2$ has read a data item that was never committed and so never really existed \cite{16, 10.1145/223784.223785}.
One specification of the Dirty Read, which is still widely used in exiting RDBMSs, is restated below:
\begin{center}
\small
$W_1(x, 5) \dots R_2(x, 5)  \dots $ ($A_1$ and $C_2$ in any order)
\end{center}

\noindent Nevertheless, the above specification is single-variable based, and cannot be extended to the multi-variable based Dirty Read.
Assume there exists a transfer operation from one account $x$ to another account $y$, where the invariant is that $x + y = 10$. Consider the following specification:
\begin{center}
\small
$R_1(x, 5) W_1(x, 2) R_2(x, 2) R_2(y, 5) C_2 R_1(y, 5) W_1(y, 8) C_1$
\end{center}

\noindent As we can see, $T_1$ makes the transfer correctly, but $T_2$ observes an intermediate state of the database, causing a \textit{Intermediate Reads} (the total is only 7).
To avoid either \textit{Aborted Reads} \cite{839388}, or Intermediate Reads \cite{839388}, or \textit{Circular Information Flow} \cite{839388}, or both, revised definitions are proposed with a more precise and complete specification. For example, Jim Grey et. al \cite{10.1145/223784.223785} formally define Dirty Read below:
\begin{center}
\small
$W_1(x, 5) \dots R_2(x, 5)  \dots $ (($A_1$ or $C_1$) and ($C_2$ or $A_2$) in any order)
\end{center}

\noindent Besides Dirty Read, the other three data anomalies are re-defined in mathematical forms in order to avoid ambiguous interpretations \cite{10.1145/223784.223785}.
Unfortunately, due to the re-definition of data anomalies, although the meanings are more precise and complete, the interpretations become more difficult because of the lack of uniform standards.

On the other hand, the definitions are made in a case-by-case manner, and lead to a situation that even a senior DBA could not have infallible knowledge of data anomalies, due to a lack of a complete understanding of its nature.
Except the four data anomalies defined in the ANSI/ISO SQL standard, typical data anomalies are in different forms, such as \textit{Read skew} \cite{10.1145/223784.223785}, \textit{Write Skew} \cite{10.1145/223784.223785}, \textit{Fractured Reads} \cite{cerone2017algebraic}, \textit{Cross-phenomenon} \cite{binnig2014distributed}. 
A thorough study on state-of-the-art data anomalies will be conducted in Section \ref{sec_dataAnomaliesModel}.
Apparently, without a complete understanding of nature, it is rather challenging to know what they are and what are the relationship among them.
Even worse, the revised definitions with the precise and complete specifications further deteriorate the understanding of data anomalies.

While revised definitions in existing literature \cite{16, 10.1145/223784.223785, 839388} propose some mathematical formalization to correct the former argument, how can we address the latter argument still remains an open problem, we need a complete form to define data anomalies instead of case by case.

In legacy database systems, ACID(Atomicity, Consistency, Isolation, Durability) refers to a standard set of properties that guarantee database transactions are processed reliably. Consistency and Isolation are two important properties, the former ensures the correctness of data, and the latter ensures the high performance of database systems. However, the latter improves the performance of the database at the expense of the consistency of the data. Weak isolation levels allow some data anomalies to occur.

\cite{ 10.1145/223784.223785} defines eight data anomalies such as Read Skew, Write Skew. \cite{DBLP:journals/tods/FeketeLOOS05} defines some other data anomalies such as \textit{Predicate-based Write Skew}. In recent years, there still exist extensive research works that focus on reporting new data anomalies, including Intermediate Read \cite{839388,xie2015high}, \textit{Serial-concurrency-phenomenon} and Cross-phenomenon \cite{binnig2014distributed}, \textit{Long Fork} anomaly and \textit{Causality Violation} anomaly \cite{cerone2017algebraic}. It seems that the conflict graph can describe data anomalies, but Adya \cite{839388} defines isolation levels with conflict graph and three specific data anomalies (Dirty Writes, Dirty Reads, Intermediate Reads), his method is not unified and can neither define what is data anomaly nor cover all data anomalies. So we wonder \textbf{how many data anomalies exist in the application?} We analyze Read Skew is similar with Serial-concurrency-phenomenon even though they are the same kind of data anomalies with different names by their formal definition (Table \ref{tab:datano}). But why do they have different names, and \textbf{is there a unified standard to study data anomalies?} Further, we do not know whether there are some new coming data anomalies that will affect the application. At the same time, we do not even precisely know: \textbf{what is the relationship between data anomalies and data consistency or isolation levels?}

For reference, we make a thorough survey on data anomalies reported in the state-of-the-art literature and list them in Table~\ref{tab:datano}. The table shows there are some different data anomalies reported. Surprisingly, some are the same anomalies but reported with different names, e.g., Long Fork and Cross-phenomenon in Table~\ref{tab:datano}. The reason is that, no formal definition to express these anomalies. Worse, it is ambiguous which current isolation level can eliminate these two data anomalies? 

In this paper, We make the following contributions.
\begin{itemize}
    \item 
We propose a general framework, called Coo, which models data anomalies by extending the conflicted relationship of the conflict graph. In our framework, we can formally define all data anomalies(\S \ref{sec_dataAnomaliesModel}). At the same time, we report 20+ new entity-based data anomalies (Table \ref{table:anomaly_classification}) and all predicate-based data anomalies (Table \ref{table:predicate_anomalies}).


\item 
We, by Coo, can specify all simplified primitive data anomalies (\S \ref{sec_ConflictGraphs}) and classify them.
We also explore the 
similarity and difference between entity-based data anomalies and predicate-based data anomalies (\S \ref{sec_ConflictRelationsStatus}).

\item We quantitatively study the data anomalies(\S \ref{sec_TestandQuantitativeDataAnomalies}) and rollback rate(\S \ref{sec_EvaluationofRollbackRate}), propose the new isolation levels (\S \ref{sec_iso}), and analyze the relationship between data anomalies, isolation levels and concurrency control algorithms(\S \ref{sec_algorithm}). Our new isolation levels can classify any types of anomalies, in contrast to the current isolation levels which are defined based on limited known data anomalies.
To the best of our knowledge, this is the first paper to systematically and quantitatively study data anomalies, rollback rates, isolation levels and concurrent control algorithms. 



\end{itemize}

The rest of this paper is organized as follows.
Section 2 presents the Coo framework. 
Section 3 discusses how to quantify data anomalies and rollback rates.
Section 4 shows the model application on databases, and discusses how to define isolation levels and how to analyze concurrent control algorithms.
Section 5 discusses related work.
Section 6 is a conclusion.


\section{Data anomalies Model}\label{sec_dataAnomaliesModel}
In this section, we first give the symbolizations, then define entity-based and predicated-based partial order pairs, and finally present our \textit{Coo} framework to formulate data anomalies.


\subsection{Abstract System Model}\label{section21}
  We consider storing objects $Obj=\{x,y,...\}$ in a concurrent system. Transactions interact with the objects by issuing read and write operations, grouped them together into transactions symbol $t_i$. We let $Op_i=\{R_i[x_n],W_i[x_n] | x \in Obj \}$ describe the possible operations invocations: reading the version $n$ from an object $x$ or writing version $n$ to $x$ by the transaction $t_i$. The read operation is divided into two categories. One type is a physical read $R_i[x_n]$, which describes the transaction $t_i$ read the object $x$ with version $n$. The other type is a predicate read operation $R_i^P[x_n]$, which describes that the object $x$ with version $n$ was in the set of objects fitting predicate logic $P$ by the transaction $t_i$, i.e, $R_i[x_n \in Vset(P)]$. A \textbf{predicate} is a sentence that contains a finite number of variables and becomes a statement when specific objects are substituted for the variables. The \textbf{domain} of a predicate $Vset(P)$ is the set of all objects that may be substituting in place of the variables \cite{epp2010discrete}. In logic, predicates can be obtained by removing some or all of the nouns $x$ or $x,y$ from a statement and be symbolized as $P(x)$ or $P(x,y)$. Always predicate is a  predicate symbol together with suitable predicate variables referred to as propositional functions or open sentences. Predicate logic includes first-order logic, second-order logic and $n$-order logic. In first-order logic, a predicate can only refer to a single subject. First-order logic is also known as first-order predicate calculus or first-order functional calculus \cite{smullyan1995first}. A sentence in first-order logic is written in the form $P(x)$, where $P$ is the predicate and $x$ is the subject, represented as a variable.  Second-order logic is a predicate logic of two predicate variables after simplie. The same definition method is for $n$-order logic. We define that \textbf{atomic formula} for an $n$-tuple $P(x,y,\dots)$ where $P$ is any predicate of degree $n$ and $x,y,\dots$ are $n$ cannot be divided symbols (or variables).Smullyan and Raymond M \cite{janin1996expressive}proved that a $n$ order logic can containing only disjunction of sentence of the form: $P(x,y,\dots) = \wedge (or \vee) P_i(x_i) \wedge(or \vee) P_j(x_j,y_j) \dots \wedge (or \vee) P_z(x_z,y_z,\dots)$. Besides, a physical read operation can be regarded as the predicate read operation with the predicate variable happens to be the object ID. 

  Therefore, anomalies of a schedule can be detected just under the condition of predicate logic of degree from one to $n$ in sequence.

  In order to reflect the affection between writing and read operations and predicate logic domain set, we extend senses of versions of the object which adding insertion and deletion versions.
  \begin{itemize}
    \item The object $x$ has an initial version $x_{init}$ called unborn version. When transaction $t_i$ creates an object that will ever exist in the system through writing operation, we denote $W_i[x_0]$. Of course, there is no read operation $R_i[x_{init}]$. 
    \item The object $x$ has a visible version $x_n(n=0,1,2,\dots)$ when a transaction updates an object to generate its new version. We use $R_i[x_n]$ to describe a transaction read a version of an object that is in the set.
    \item The object $x$ has a dead version $x_{dead}$ when the transaction deletes an object $x$. We use $W_i[x_{dead}]$ to describe this situation. Also, there is no read operation $R_i[x_{dead}]$. 
  \end{itemize}

  Based on the definition of versions of an object, there are four cases of reading and writing operations under the predicate logic.
  \begin{itemize}
    \item The version of the object $x_n$ is in the predicate domain set $Vset(P)$ by the predicated-based read operation, which denoted $R_i[x_n \in Vset(P)]$.
    \item  The version of the object $x_n$ is not in the predicate domain set $Vset(P)$ by the predicate-based read operation, which denoted $R_i[x_n \notin Vset(P)]$.
    \item The version of the object $x_n$ is in the predicate domain set $Vset(P)$ by the predicate-based write operation, which denoted $W_i[x_n \in Vset(P)]$.
    \item The version of the object $x_n$ is not in the predicate domain set $Vset(P)$ by the predicate-based write operation, which denoted $W_i[x_n \notin Vset(P)]$.
  \end{itemize}
  We write $W_i[x_n]$ indicated writing a value $x_n$ that is irrelevant to whether or not in the predicate domain set $Vset(P)$.

  According to the extended definition of object versions, $W_i[x_0 \in Vset(P)]$ indicates that the transaction $t_i$ inserts the new object exactly in the predicate set. And  $W_i[x_{dead} \in Vset(P)]$ indicates that the transaction $t_i$ deleted the object exactly in the predicate set.
  
 \subsection{Conflict Relations}\label{section22}
 We define the different types of predicate-based operation conflicts that can occur to a concurrent system in this section. We define two types of conflict relations and describe in detail different operations combination forms of various categories of conflict relations.
\begin{definition}\label{def:EdgeType}
  Let $t_i, t_j$ be transactions in a schedule $s$, $t_i <_s t_j$. Two operations $p\in t_i$ and $q \in t_j$ access the same object and at least one of them is a write, i.e.,$\{ W_i W_j[x],W_i R_j[x],R_i W_j[x]\}$. 
\end{definition}

 We separate the above relations to several cases by whether the object read of predicate-based domain.

\subsubsection{Write-Write Dependency}
\begin{definition}
  Write-write conflicts $WW[x]$ occur when one transaction overwrites a version written by another transaction.
\end{definition}

Since predicate-based operations are queries objects sets. There is no notion of predicate-based write-write conflicts.

\subsubsection{Read-Write Dependency}
\begin{definition}
Read-write conflicts $RW[x]$ occur when one transaction reads or predicate-based $P$ read a relevant object version by some other transaction.
\end{definition}

There are four situations based on the relationship between the object version and the predicate domain set $Vset(P)$.
\begin{itemize}
  \item[1] $R_i[x_n \in Vset(P)]W_j[x_{n+1} \in Vset(P)]$
  \item[2] $R_i[x_n \in Vset(P)]W_j[x_{n+1} \notin Vset(P)]$
  \item[3] $R_i[x_n \notin Vset(P)]W_j[x_{n+1} \in Vset(P)]$
  \item[4] $R_i[x_n \notin Vset(P)]W_j[x_{n+1} \notin Vset(P)]$
\end{itemize}

Obviously, case 1 belongs to entity-based conflicts. Case 2 belongs to predicate-based conflicts because after reading the object $x$, the write operation gets $x$ out of the predicated domain. Case 3 belongs to predicate-based conflicts, in which the write operation of transaction $t_j$ inserts the object of the predicate set through the insert or update. In case 4, the object is not visible read by transaction $t_i$, and the write operation of transaction $t_j$ does not produce an intersection with the predicate set $Vset(P)$, so there is no relationship between two operations.

We formally express the read-write conflicts as
$$
\begin{aligned}
   R_iW_j[x] =&\{ R_i[x_n \in Vset(P)]W_j[x_{n+1} \in Vset(P)], \\ 
   & R_i[x_n \in Vset(P)]W_j[x_{n+1} \notin Vset(P)],\\
   &R_i[x_n \notin Vset(P)]W_j[x_{n+1} \in Vset(P)] \} 
\end{aligned}
$$

\subsubsection{Write-Read Dependency}
\begin{definition}
  Write-read conflicts $WR[x]$ occur when a transaction overwrites a version observed or predicate-based $P$ observed by some other transaction.
\end{definition}

There are two situations based on the relation between the object version and the predicate domain set $Vset(P)$.

\begin{itemize}
  \item[1] $W_i[x_n \in Vset(P)]R_j[x_n \in Vset(P)]$
  \item[2] $W_i[x_n \notin Vset(P)]R_j[x_n \notin Vset(P)]$
\end{itemize}

Thus the read-write conflicts can be expressed as:
$$
\begin{aligned}
   W_iR_j[x] = &\{ W_i[x_n \in Vset(P)]R_j[x_n \in Vset(P)], \\ 
   & W_i[x_n \notin Vset(P)]R_j[x_n \notin Vset(P)] \} 
\end{aligned}
$$

We divide the conflict relationship into entity conflicts and predicate conflicts:
\begin{itemize}
    \item []$\begin{aligned}
   \text{Entity Conflicts} = &\{ W_i[x_n] W_j[x_{n+1}], \\ 
   & R_i[x_n \in Vset(P)]W_j[x_{n+1} \in Vset(P)],\\
   & W_i[x_n \in Vset(P)]R_j[x_n \in Vset(P)]\} ;
\end{aligned}
$
\item []$\begin{aligned}
   \text{Predicate Conflicts} = &\{ R_i[x_n \in Vset(P)]W_j[x_{n+1} \notin Vset(P)], \\ 
   &  R_i[x_n \in Vset(P)]W_j[x_{n+1} \notin Vset(P)],\\
   & W_i[x_n \notin Vset(P)]R_j[x_n \notin Vset(P)]\} .
\end{aligned}
$
\end{itemize}
Based on the above discussion, We can still symbolize the conflict relations as $conf = \{ W_iW_j,W_iR_j,R_iW_j, \}$. 

\subsection{Conflict Graphs}\label{sec_ConflictGraphs}

\begin{definition} (\textbf{Conflict Equivalent})
  Let $s$ and $s'$ be two schedules. $s$ and $s'$ are called conflict equivalent, denoted $s \thickapprox s' $, if they have the same conflict relations, including the entity-based and the predicate-base, i.e., if the following holds:
  \begin{itemize}
    \item[1] $Op(s)=Op(s')$;
    \item[2] $conf(s) = conf(s')$.
  \end{itemize}
\end{definition}

Thus, two schedules are conflict equivalents if all conflicting pairs of steps from distinct transactions occur in the same order in both schedules. The conflict relations of a schedule can be described in terms of graphs \cite{2002Concurrency}. 
\begin{definition}(\textbf{Conflict Graph})
  Let $s$ be a schedule and $T(s)$ be a set of transactions belong to schedule $s$. A graph $G(s)=(V,E)$ is called conflict graph, if vertices are transactions set and edges are operations with conflict relations, i.e., if the following holds:
  \begin{itemize}
    \item $V \subset T(s)$;
    \item $(p_i q_j) \in E \Leftrightarrow t_i \neq t_j \wedge (p_i q_j) \in conf(s)$.
  \end{itemize}
\end{definition}

We are now ready to introduce a notion of anomalies and consistency in concurrency systems.
\begin{definition}(\textbf{Data Anomalies})\label{DEFanomalies}
  A schedule $s$ has anomalies if there exists a cycle in the conflict graph.
\end{definition} 

\begin{definition}\label{DEFanomaliesType}
We call an anomaly as an \textbf{entity-based} if the anomaly cycle only has entity conflicts. If there is at least one kind of predicate conflicts, we call it as \textbf{predicate-based} anomaly. 
\end{definition}


\begin{example}
  If 
$$
\begin{aligned}
  s = &W_1[x_0 \in Vset(P)]W_2[x_1 \notin Vset(P)]W_2[y_1 \in Vset(P)] \\
  &W_3[y_2 \in Vset(P)]R_1[x_1 \notin Vset(P)]R_1[y_2 \in Vset(P)]C_1C_2
\end{aligned}
$$
where $t_1$ and $t_2$ are committed, and $t_3$ still active.

The transaction $t_1$ inserts an object $x$ in predicate domain set $Vset(P)$, but the transaction $t_2$ updates it and move it out. The transaction $t_3$ updates the object $y$ and the transaction $t_1$ read it.

We have 
$$ 
\begin{aligned}
\text{Entity conflicts} = &\{W_1 W_2[x],W_2W_3[y],\\
                         &W_3[y_2 \in Vset(P)]R_1[y_2 \in Vset(P)]\}.
\end{aligned}
$$
$$
\text{Predicate conflict} = \{W_2[x_1 \notin Vset(P)]R_1[x_1 \notin Vset(P)]\}.
$$

\end{example}

\begin{example}
  A schedule has two predicate logics and $Vset(P_1)\bigcap Vset(P_2)=\phi$,  i.e.,
  $$
  \begin{aligned}
    s = &W_1[x_0 \in Vset(P_1)]W_2[x_1 \notin Vset(P_2)]W_2[y_1 \in Vset(P_2)] \\
    &W_3[y_2 \in Vset(P_1)]R_1[x_1 \notin Vset(P_2)]R_1[y_1 \in Vset(P_1)]C_1C_2.
  \end{aligned}
  $$
  Firstly, we rewrite the schedule with visible read under non-predicate logic
  $$
  s_{1} = W_1[x_0]W_2[x_1]W_2[y_1]W_3[y_2]R_1[y_1].
  $$
  We have the entity conflicts set
  $$
  \text{Entity conflict} = \{W_1 W_2[x],W_2W_3[y],W_3R_1[y]\}.
  $$

  The only invisible read operation is $R_1[x_1 \notin Vset(P_2)]$. We choice operations with the predicate $P_2$
  $$
  \text{Predicate conflict} = \{W_2[x_1 \notin Vset(P_2)]R_1[x_1 \notin Vset(P_2)]\}.
  $$
\end{example}
So the conflict graph is shown as fig.\ref{fig1}.
\begin{figure}[h]
  \centering
  \includegraphics[width= 0.7\linewidth]{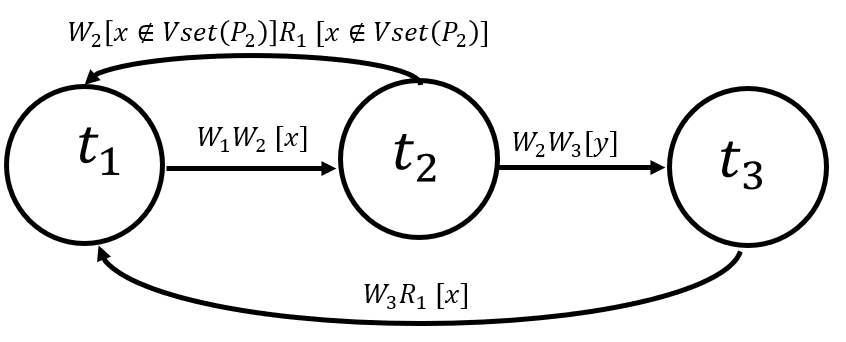}
  \caption{Conflict graphs of example 2.}
  \label{fig1}
\end{figure}

According to the definition of consistency, there may be multiple cycles in a conflict graph. Just one cycle can be a crucial disturbance in the system.  Therefore, we specify the earliest cycle that appeared in the conflict graph to be studied.  We prove that a conflict cycle constructed by an object at most has two transactions, and a conflict cycle constructed by $n$ objects at most has $2n$ transactions, and we use edges to classify them. 

\begin{definition}(\textbf{Data Anomalies Types}) \label{Types}
    In the conflict graph formed by conflicts, we separate into three types of anomalies based on the operations on the edges.
  \begin{itemize}
    \item[A1:] We defined a conflict graph with a cycle as \textbf{Read Anomaly Type, RAT}, if the graph has a $WR$ edge at least, i.e., for a graph $G=\{V, E\}$ with $ \exists WR\in E$.
    \item[A2:]  We defined a conflict graph with a cycle as \textbf{Write Anomaly Type, WAT}, if the graph does not contain $WR$ edges, but contain at least a $WW$ edge, i.e.,for a graph $G=\{V,E\}$ with $ (\forall WR \notin E)\wedge (\exists WW\in E) $.  
    \item[A3:]  We defined a conflict graph with a cycle as \textbf{Intersect Anomaly Type, IAT}, if the graph does not contain $WW$ edges, also does not contain $WR$ edges, i.e.,for a graph $G=\{V,E\}$ with $ (\forall WW \notin E)\wedge (\forall WR \notin E) $.
  \end{itemize}
\end{definition}

\begin{table*}[ht]
\caption{Combinations of bilateral cycles with predicate relations.} 
\footnotesize

\begin{tabular}{|c|l|c|l|l|}
\hline
\textbf{Types of   Anomalies} & \multicolumn{1}{c|}{\textbf{Edges   Combanations}} & \textbf{Anomalies}                   & \multicolumn{1}{c|}{\textbf{Partially   ordered combination with predicate}}      & \multicolumn{1}{c|}{\textbf{Classification}} \\ \hline
SDA                           & $W-W-W[x]$                                         & Full-write                           & $W-W-W[x]$                                                                        & Entity-based                                 \\ \hline
\multirow{2}{*}{SDA}          & \multirow{2}{*}{$W-W-R[x]$}                        & \multirow{2}{*}{Lost Self Update}    & $W_i[x]-W_j[x \in   Vset(P)]-R_i[x \in Vset(P)]$                                  & Entity-based                                 \\ \cline{4-5} 
                              &                                                    &                                      & $W_i[x]-W_j[x \notin Vset(P)]-R_i[x \notin Vset(P)]$                              & Predicate-based                              \\ \hline
\multirow{3}{*}{SDA}          & \multirow{3}{*}{$W-R-W[x]$}                        & \multirow{3}{*}{Intermediate Read}   & $W_i[x \in   Vset(P)]-R_j[x \in Vset(P)]-W_i[x \in Vset(P)]$                      & Entity-based                                 \\ \cline{4-5} 
                              &                                                    &                                      & $W_i[x \in Vset(P)]-R_j[x \in Vset(P)]-W_i[x \notin Vset(P)]$                     & Predicate-based                              \\ \cline{4-5} 
                              &                                                    &                                      & $W_i[x \notin Vset(P)]-R_j[x \notin Vset(P)]-W_i[x \in Vset(P)]$                  & Predicate-based                              \\ \hline
\multirow{3}{*}{SDA}          & \multirow{3}{*}{$R-W-W[x]$}                        & \multirow{3}{*}{Lost Update}         & $R_i[x \in   Vset(P)]-W_j[x \in Vset(P)]-W_i[x]$                                  & Entity-based                                 \\ \cline{4-5} 
                              &                                                    &                                      & $R_i[x \in Vset(P)]-W_j[x \notin Vset(P)]-W_i[x]$                                 & Predicate-based                              \\ \cline{4-5} 
                              &                                                    &                                      & $R_i[x \notin Vset(P)]-W_j[x \in Vset(P)]-W_i[x]$                                 & Predicate-based                              \\ \hline
\multirow{3}{*}{SDA}          & \multirow{3}{*}{$R-W-R[x]$}                        & \multirow{3}{*}{Non-repeatable Read} & $R_i[x \in   Vset(P)]-W_j[x \in Vset(P)]-R_i[x \in Vset(P)]$                      & Entity-based                                 \\ \cline{4-5} 
                              &                                                    &                                      & $R_i[x \notin Vset(P)]-W_j[x \in Vset(P)]-R_i[x \in Vset(P)]$                     & Predicate-based                              \\ \cline{4-5} 
                              &                                                    &                                      & $R_i[x \in Vset(P)]-W_j[x \notin Vset(P)]-R_i[x \notin Vset(P)]$                  & Predicate-based                              \\ \hline
DDA                           & $WW[x]-WW[y]$                                      & Full-write   Skew                    & $WW[x]-WW[y]$                                                                     & Entity-based                                 \\ \hline
\multirow{2}{*}{DDA}          & \multirow{2}{*}{$WW[x]-WR[y]$}                     & \multirow{2}{*}{Double-write Skew 2} & $WW[x]-W[y \in   Vset(P)]R[y \in Vset(P)]$                                        & Entity-based                                 \\ \cline{4-5} 
                              &                                                    &                                      & $WW[x]-W[y \notin Vset(P)]R[y \notin Vset(P)]$                                    & Predicate-based                              \\ \hline
\multirow{3}{*}{DDA}          & \multirow{3}{*}{$WW[x]-RW[y]$}                     & \multirow{3}{*}{Read-Write Skew 2}   & $WW[x]-R[y \in   Vset(P)]W[y \in Vset(P)]$                                        & Entity-based                                 \\ \cline{4-5} 
                              &                                                    &                                      & $WW[x]-R[y \notin Vset(P)]W[y \in Vset(P)]$                                       & Predicate-based                              \\ \cline{4-5} 
                              &                                                    &                                      & $WW[x]-R[y \in Vset(P)]W[y \notin Vset(P)]$                                       & Predicate-based                              \\ \hline
\multirow{2}{*}{DDA}          & \multirow{2}{*}{$WR[x]-WW[y]$}                     & \multirow{2}{*}{Double-Write Skew 1} & $W[x \in Vset(P)]R[x   \in Vset(P)]-WW[y]$                                        & Entity-based                                 \\ \cline{4-5} 
                              &                                                    &                                      & $W[x \notin Vset(P)]R[x \notin Vset(P)]-WW[y]$                                    & Predicate-based                              \\ \hline
\multirow{4}{*}{DDA}          & \multirow{4}{*}{$WR[x]-WR[y]$}                     & \multirow{4}{*}{Write-Read Skew}     & $W[x \in Vset(P)]R[x   \in Vset(P)]-W[y \in Vset(P)]R[y \in Vset(P)]$             & Entity-based                                 \\ \cline{4-5} 
                              &                                                    &                                      & $W[x \in Vset(P)]R[x \in Vset(P)]-W[y \notin Vset(P)]R[y \notin   Vset(P)]$       & Predicate-based                              \\ \cline{4-5} 
                              &                                                    &                                      & $W[x \notin Vset(P)]R[x \notin Vset(P)]-W[y \in Vset(P)]R[y \in   Vset(P)]$       & Predicate-based                              \\ \cline{4-5} 
                              &                                                    &                                      & $W[x \notin Vset(P)]R[x \notin Vset(P)]-W[y \notin Vset(P)]R[y \notin   Vset(P)]$ & Predicate-based                              \\ \hline
\multirow{6}{*}{DDA}          & \multirow{6}{*}{$WR[x]-RW[y]$}                     & \multirow{6}{*}{Read Skew 2}         & $W[x \in Vset(P)]R[x   \in Vset(P)]-R[y \in Vset(P)]W[y \in Vset(P)]$             & Entity-based                                 \\ \cline{4-5} 
                              &                                                    &                                      & $W[x \in Vset(P)]R[x \in Vset(P)]-R[y \in Vset(P)]W[y \notin Vset(P)]$            & Predicate-based                              \\ \cline{4-5} 
                              &                                                    &                                      & $W[x \in Vset(P)]R[x \in Vset(P)]-R[y \notin Vset(P)]W[y \in Vset(P)]$            & Predicate-based                              \\ \cline{4-5} 
                              &                                                    &                                      & $W[x \notin Vset(P)]R[x \notin Vset(P)]-R[y \in Vset(P)]W[y \in   Vset(P)]$       & Predicate-based                              \\ \cline{4-5} 
                              &                                                    &                                      & $W[x \notin Vset(P)]R[x \notin Vset(P)]-R[y \in Vset(P)]W[y \notin   Vset(P)]$    & Predicate-based                              \\ \cline{4-5} 
                              &                                                    &                                      & $W[x \notin Vset(P)]R[x \notin Vset(P)]-R[y \notin Vset(P)]W[y \in   Vset(P)]$    & Predicate-based                              \\ \hline
\multirow{3}{*}{DDA}          & \multirow{3}{*}{$RW[x]-WW[y]$}                     & \multirow{3}{*}{Read-Write Skew 1}   & $R[x \in Vset(P)]W[x   \in Vset(P)]-WW[y]$                                        & Entity-based                                 \\ \cline{4-5} 
                              &                                                    &                                      & $R[x \in Vset(P)]W[x \notin Vset(P)]-WW[y]$                                       & Predicate-based                              \\ \cline{4-5} 
                              &                                                    &                                      & $R[x \notin Vset(P)]W[x \in Vset(P)]-WW[y]$                                       & Predicate-based                              \\ \hline
\multirow{6}{*}{DDA}          & \multirow{6}{*}{$RW[x]-WR[y]$}                     & \multirow{6}{*}{Read Skew}           & $R[x \in Vset(P)]W[x   \in Vset(P)]-W[y \in Vset(P)]R[y \in Vset(P)]$             & Entity-based                                 \\ \cline{4-5} 
                              &                                                    &                                      & $R[x \in Vset(P)]W[x \notin Vset(P)]-W[y \in Vset(P)]R[y \in Vset(P)]$            & Predicate-based                              \\ \cline{4-5} 
                              &                                                    &                                      & $R[x \notin Vset(P)]W[x \in Vset(P)]-W[y \in Vset(P)]R[y \in Vset(P)]$            & Predicate-based                              \\ \cline{4-5} 
                              &                                                    &                                      & $R[x \in Vset(P)]W[x \in Vset(P)]-W[y \notin Vset(P)]R[y \notin   Vset(P)]$       & Predicate-based                              \\ \cline{4-5} 
                              &                                                    &                                      & $R[x \in Vset(P)]W[x \notin Vset(P)]-W[y \notin Vset(P)]R[y \notin   Vset(P)]$    & Predicate-based                              \\ \cline{4-5} 
                              &                                                    &                                      & $R[x \notin Vset(P)]W[x \in Vset(P)]-W[y \notin Vset(P)]R[y \notin   Vset(P)]$    & Predicate-based                              \\ \hline
\multirow{9}{*}{DDA}          & \multirow{9}{*}{$RW[x]-RW[y]$}                     & \multirow{9}{*}{Write Skew}          & $R[x \in Vset(P)]W[x   \in Vset(P)]-R[y \in Vset(P)]W[y \in Vset(P)]$             & Entity-based                                 \\ \cline{4-5} 
                              &                                                    &                                      & $R[x \in Vset(P)]W[x \notin Vset(P)]-R[y \in Vset(P)]W[y \in Vset(P)]$            & Predicate-based                              \\ \cline{4-5} 
                              &                                                    &                                      & $R[x \notin Vset(P)]W[x \in Vset(P)]-R[y \in Vset(P)]W[y \in Vset(P)]$            & Predicate-based                              \\ \cline{4-5} 
                              &                                                    &                                      & $R[x \in Vset(P)]W[x \in Vset(P)]-R[y \notin Vset(P)]W[y \in Vset(P)]$            & Predicate-based                              \\ \cline{4-5} 
                              &                                                    &                                      & $R[x \in Vset(P)]W[x \notin Vset(P)]-R[y \notin Vset(P)]W[y \in   Vset(P)]$       & Predicate-based                              \\ \cline{4-5} 
                              &                                                    &                                      & $R[x \notin Vset(P)]W[x \in Vset(P)]-R[y \notin Vset(P)]W[y \in   Vset(P)]$       & Predicate-based                              \\ \cline{4-5} 
                              &                                                    &                                      & $R[x \in Vset(P)]W[x \in Vset(P)]-R[y \in Vset(P)]W[y \notin Vset(P)]$            & Predicate-based                              \\ \cline{4-5} 
                              &                                                    &                                      & $R[x \in Vset(P)]W[x \notin Vset(P)]-R[y \in Vset(P)]W[y \notin   Vset(P)]$       & Predicate-based                              \\ \cline{4-5} 
                              &                                                    &                                      & $R[x \notin Vset(P)]W[x \in Vset(P)]-R[y \in Vset(P)]W[y \notin   Vset(P)]$       & Predicate-based                              \\ \hline
\end{tabular}

\label{table:predicate_anomalies}
\end{table*}

\subsection{Simplify of Conflict cycles}\label{sec_SimplifyOfConflictCycles}
\begin{theorem} \label{theorem1}
  If a single object from 3-transactions constitutes a cycle, the conflict graph can be reduced to a cycle between two transactions. 
\end{theorem}
\begin{proof}
  Suppose the conflict graph is $G = \{ \{t_1,t_2,t_3\},\{(p_1q_2[x]),\\(p_2q_3[x]),(p_3q_1[x])\}\}$ which $(pq)$ be one of $\{ WW,WR,RW\}$. 
  
  When $p_1=W$ in $(p_1q_2[x])$, no matter what operation $p_2$ (W or R) is, there will be a conflict relationship $(p_1p_2[x])$ or $(p_2p_1[x])$. If $(p_2p_1[x])$, there is a two transactions graph $G' =\{ \{t_1,t_2\},\{(p_2p_1[x]), \\ (p_1q_2[x])\} \}$. If $(p_1p_2[x])$, because of existing the conflict relationship $(p_2q_3[x])$, then we have $p_1 <_s p_2 <_s q_3$, which means that it has the conflict relationship $(p_1q_3)$. So we get the two transaction graph $G' =\{ \{t_1,t_3\},\{(p_1q_3[x]),(p_3q_1[x])\} \}$
  
  When $p_1=R$ in $(p_1q_2[x])$, $q_2$ can only be $W$ to constitute the conflict. No matter what operation $p_3$ (W or R) in $(p_3q_1[x])$ is, there will be a conflict relationship $(q_2p_3[x])$ or $(p_3q_2[x])$. If $(q_2p_3[x])$, because of existing the conflict relationship $(p_3q_1[x])$, then we have $q_2 <_s p_3 <_s q_1$, which means that it has the conflict relationship $(q_2q_1)$. So we get the two transaction graph $G' =\{ \{t_1,t_2\},\{(p_1q_2[x]),(q_2q_1[x])\} \}$. If $(p_3q_2[x])$, it exists the conflict relationship $(p_2q_3[x])$. So we get the two transaction graph $G' =\{ \{t_2,t_3\},\{(p_2q_3[x]),(p_3q_2[x])\} \}$.
  \end{proof}
  \begin{example}
    The 3-transactions conflict cycle can be reduced to the 2-transactions ones, and the reduction process is shown in Figure \ref{example 3-SDA}. 
    \begin{figure}[h]
      \centering
      \includegraphics[width= 1\linewidth]{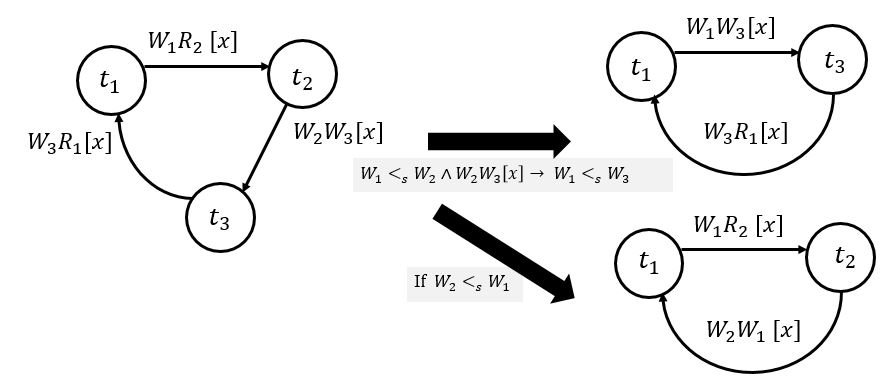}
      \caption{Reduction processes of 3-transactions conflict cycle.}
      \label{example 3-SDA}
    \end{figure}
  \end{example}

\begin{theorem} \label{theorem2}
  The multi-transaction $(\geq 3)$ conflict cycle with a single object can be reduced to a 2-transactions cycle.
\end{theorem}
\begin{proof}
  Suppose the conflict graph is $G = \{ \{t_1,t_2,\dots t_n\},\\ \{(p_1q_2[x]),(p_2q_3[x]),\dots ,(p_n q_1[x])\}\}$ which $(pq)$ also be one of $\{ WW,WR,RW\}$. We have proved $n=3$ in theorem \ref{theorem1}.
  
  Suppose when $N_T<n$ the theorem is true. When $N_T=n$, edges of the conflict cycle is $E(n)=(p_1q_2),(p_2q_3),\dots,(p_{n-1} q_n),(p_n q_1)$. 

  When $p_1=W$ in $(p_1q_2[x])$, no matter what operation $p_{n-1}$ (W or R) is, there will be a conflict relationship $(p_1p_{n-1}[x])$ or $(p_{n-1}p_1[x])$.If $(p_1p_{n-1}[x])$, because of existing the conflict relationship $(p_{n-1}q_n[x])$, then we get a 3-transaction conflict graph $G'=\{ \{t_1,t_{n-1},t_n\},\{(p_1p_{n-1}[x]), (p_{n-1}q_{n}[x]),(p_{n}q_1[x])\}\}$. According to the theorem \ref{theorem1}, $G'$ can be reduced to a 2-transactions conflict cycle.If $(p_{n-1}p_1[x])$, because of existing the conflict relationship $(p_{n-2}q_{n-1}[x])$, then we get a $n-1$-transaction conflict graph $G'=\{ \{t_1,t_2,\dots ,t_{n-1}\},\{(p_1q_2[x]),$$ \dots,$$(p_{n-2}q_{n-1}[x]),$ $(p_{n-1}p_1[x])\}\}$. According to the assumption, the graph $G'$ can be reduced to a 2-transaction conflict cycle.
 
  When $p_1=R$ in $(p_1q_2[x])$, $q_2$ can only be $W$ to constitute the conflict. No matter what operation $p_n$ (W or R) in $(p_n q_1[x])$ is, there will be a conflict relationship $(q_2 p_n[x])$ or $(p_n q_2[x])$. If $(q_2p_n[x])$, because of existing conflict relationships $(p_1q_2[x])$, then we have $p_1 <_s q_2 <_s p_n$, which means that it has the conflict relationship $(p_1 p_n)$. So we get the two transaction graph $G' =\{ \{t_1,t_n\},\{(p_1 p_n[x]),(p_n q_1[x])\} \}$. If $(p_n q_2[x])$, because of existing conflict relationships $(p_n q_1[x])$, then we get a $n-1$-transaction conflict graph $G'=\{ \{t_2,t_3,\dots ,t_{n}\},\{(p_2q_3[x]),\dots, (p_{n-1}q_{n}[x]),\\ (p_n q_2[x])\}\}$. According to the assumption, the graph $G'$ can be reduced to a 2-transaction conflict cycle.
\end{proof}
\begin{example}
  The 5-transactions conflict cycle can be reduced to the 2-transactions conflict cycle, and the reduction process is shown in Figure \ref{example 5-SDA}. 
    \begin{figure}[h]
      \centering
      \includegraphics[width= 1\linewidth]{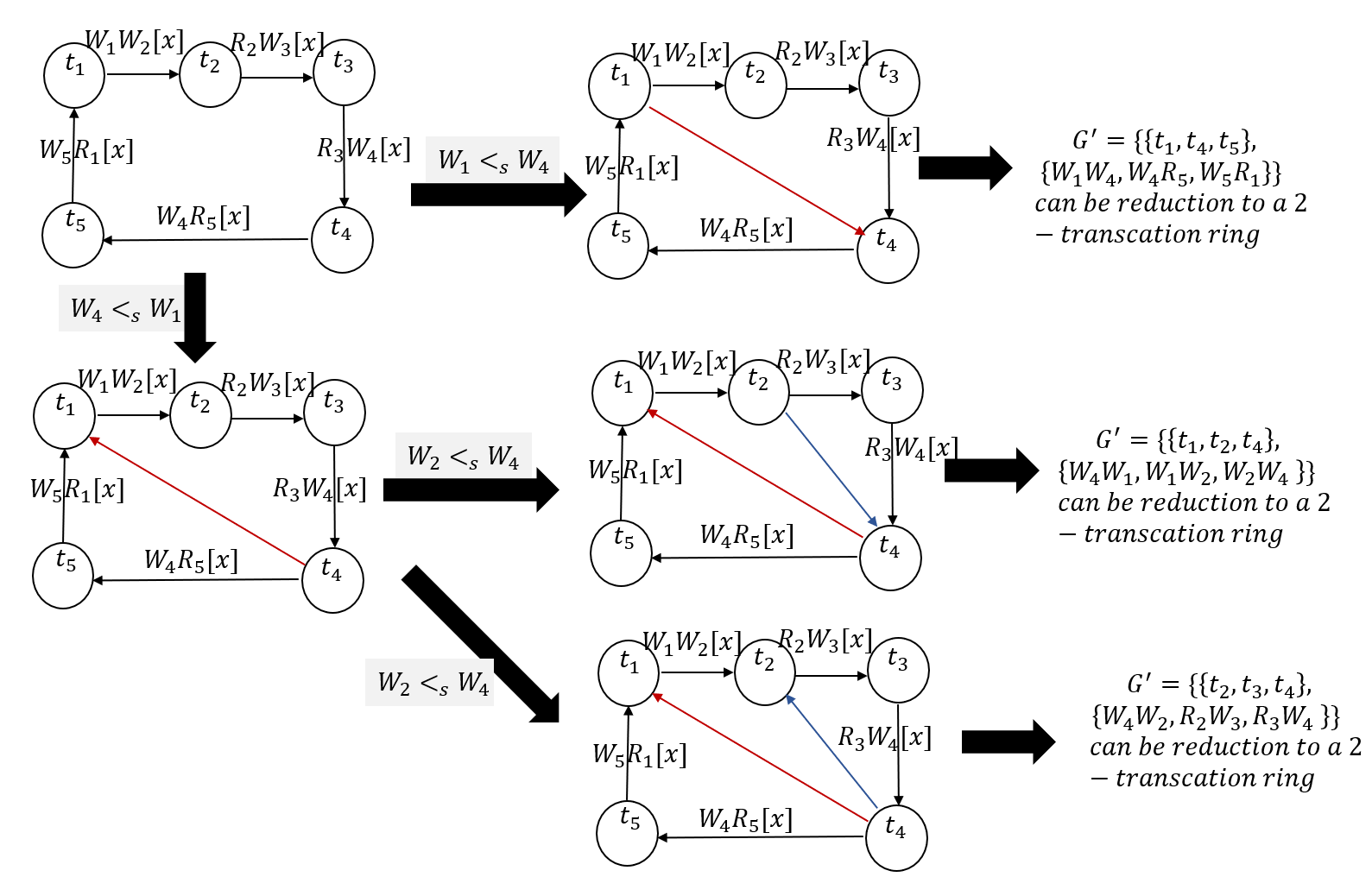}
      \caption{Reduction processes of 5-transactions conflict cycle.}
      \label{example 5-SDA}
    \end{figure}
\end{example}
According to Theorems \ref{theorem1} and \ref{theorem2}, if only one object in the conflict cycle, it can be reduced to a 2-transactions cycle. Therefore, a univariate circle can easily detect a three-operation circle $O_i -O_j-O_i$. 

\begin{theorem}\label{theorem3}
 If there are $n$ adjacent edges in conflict graph operation on the same variable, it can be reduced to one or two adjacent edges,  otherwise it will form a 2-transactions conflicts cycle by itself.  
\end{theorem}

\begin{proof}
  Suppose the conflict graph if $G={V,E}$, where $V$ is the set of transactions have conflict relations $V={t_1,t_2,\dots,t_n}$, and $E$ is the set of conflict relations $E=\{(p_1q_2),(p_2q_3),\dots, (p_{n-1}q_n)\}[x]$.
  
  When $p_1 = W$ in $(p_1q_2)$, no matter what the operation $p_{n-1}$(W or R) is, there is a conflict relationship$(p_1p_{n-1})$ or $(p_{n-1} p_1)$. If $(p_1p_{n-1})$, we have $p_1<_s p_{n-1}<_s q_n$ by the existing conflict relation $(p_{n-1} q_n)$. Thus these adjacent edges in the conflict graph can be reduced to one adjacent edges $E = \{(p_1 q_n) \}$. If $(p_{n-1}p_1)$, it will be a cycle graph $G' =\{V',E'\}$ with $E'=\{(p_{n-1}p_1),(p_1q_2),\dots, \\ (p_{n-2}q_{n-1})\}$. Thus the graph $G'$ can be reduced to a 2-transactions cycle which is proved by theorem \ref{theorem2}.
  
   When $p_1 = R$ in $(p_1q_2)$, $q_2$ only can be $W$ to constitute the conflict relation. no matter what the operation $p_{n-1}$(W or R) is, there is a conflict relationship$(q_2 p_{n-1})$ or $(p_{n-1} q_2)$. If $(q_2 p_{n-1})$, we have $q_2<_s p_{n-1}<_s q_n$ by the existing conflict relation $(q_2 q_n)$. Thus these adjacent edges in the conflict graph can be reduced to two adjacent edges $E = \{(p_1 q_2),(q_2, p_n) \}$. If $(p_{n-1} q_2)$, it will be a cycle graph $G' =\{V',E'\}$ with $E'=\{(p_{n-1}q_2),(p_2 q_3),\dots, \\ (p_{n-2} q_{n-1})\}$. Thus the graph $G'$ can be reduced to a 2-transactions cycle which is proved by theorem \ref{theorem2}.
\end{proof}
 
\begin{example}
  As is shown in Figure \ref{theropic3}, there are $3$ adjacent edges acting on the same object $x$ in the conflict graph.  It can be reduced to two adjacent edges or a 2-transactions cycle graph.
     \begin{figure}[h]
      \centering
      \includegraphics[width= 0.8\linewidth]{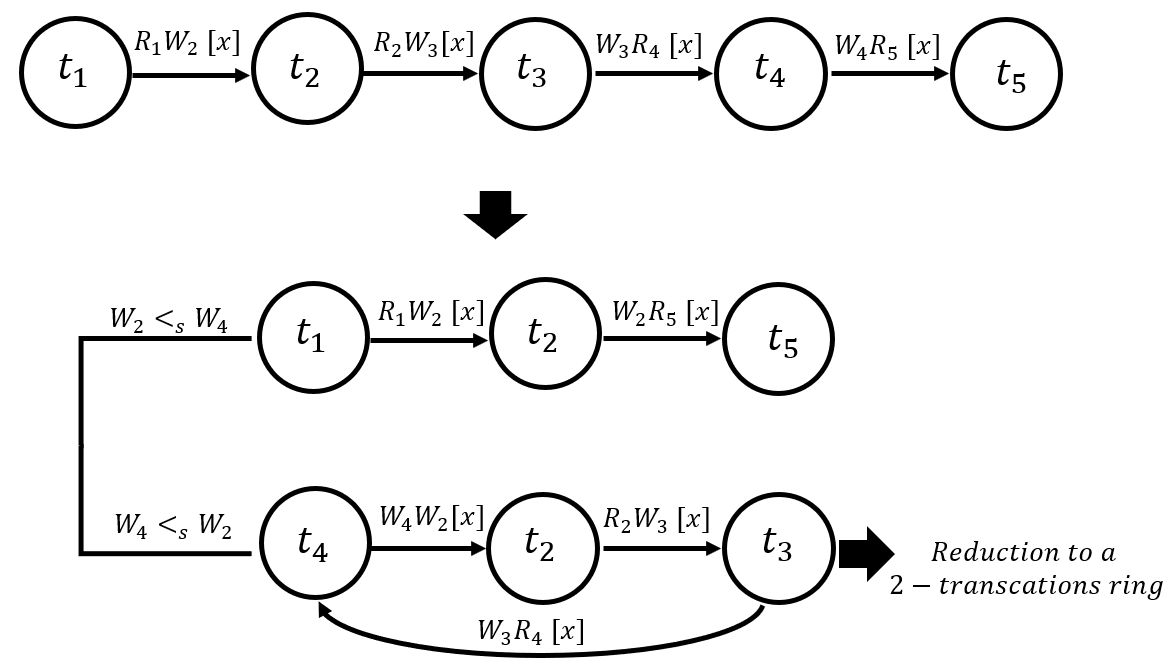}
      \caption{Reduction processes of 5-transactions adjacent edges.}
      \label{theropic3}
     \end{figure}
\end{example}

\begin{theorem}\label{theorem4}
  If there are $2$ edges of the same object with gaps in the conflict cycle, it can be reduced to one or two adjacent edges in the cycle graph.
\end{theorem}
\begin{proof}
  Suppose the conflict cycle graph is $G=\{ V, E\}$, where $V$ is the set of transactions $V=\{ t_1,t_2,\dots, t_n \}$, and $E$ is the set of conflict relations $E = \{ (p_1q_2[obj_1]),(p_2q_3[obj_2]), \dots, (p_n q_1[obj_n]) \}$. If there are two partial edges with gaps on the same object $(p_i q_j[obj])$ and $(p_s q_t[obj])$ for $j<s-1<s$.
  
  When $p_i = W$ in $(p_i q_j)$, no matter what the operation $p_{s}$(W or R) is, there is a conflict relationship$(p_i p_{s})$ or $(p_{s} p_i)$. If $(p_i p_{s})$, we have $p_i<_s p_{s}<_s q_t$ by the existing conflict relation $(p_{s} q_t)$. Thus edges in the conflict graph can be reduced to one adjacent edges $E = \{\dots, (p_i q_n)[obj], \dots \}$. If $(p_{s} p_i)$, it will be a cycle graph $G' =\{V',E'\}$ with $E'=\{(p_{s} p_i)[obj],(p_i q_j)[obj],\dots, (p_{s-1}q_{s})[obj_{s-1}]\}$. Thus the two edges with gaps are merged together.
  
   When $p_i = R$ in $(p_i q_j)$, $q_j$ only can be $W$ to constitute the conflict relation. no matter what the operation $p_s$(W or R) is, there is a conflict relationship$(q_2 p_s)$ or $(p_s q_2)$. If $(q_2 p_s)$, we have $q_2<_s p_s<_s q_t$ by the existing conflict relation $(q_2 q_t)$. Thus gaps edges in the conflict graph can be reduced to two adjacent edges $E = \{(p_i q_j),(q_j, p_t) \}$. If $(p_s q_j)$, it will be a cycle graph $G' =\{V',E'\}$ with $E'=\{(p_s q_2)[obj],(p_2 q_3)[obj_2],\dots, (p_{s-1} q_s)[obj_{s-1}]\}$. Thus the two edges with gaps are merged together.
\end{proof}

\begin{example}
  We simplify the 3-objects and 5-transactions conflict cycle, and the procession is shown in Figure \ref{theorpic4}. We reduce gaps about edges about object $x$.
     \begin{figure}[h]
      \centering
      \includegraphics[width= 0.65\linewidth]{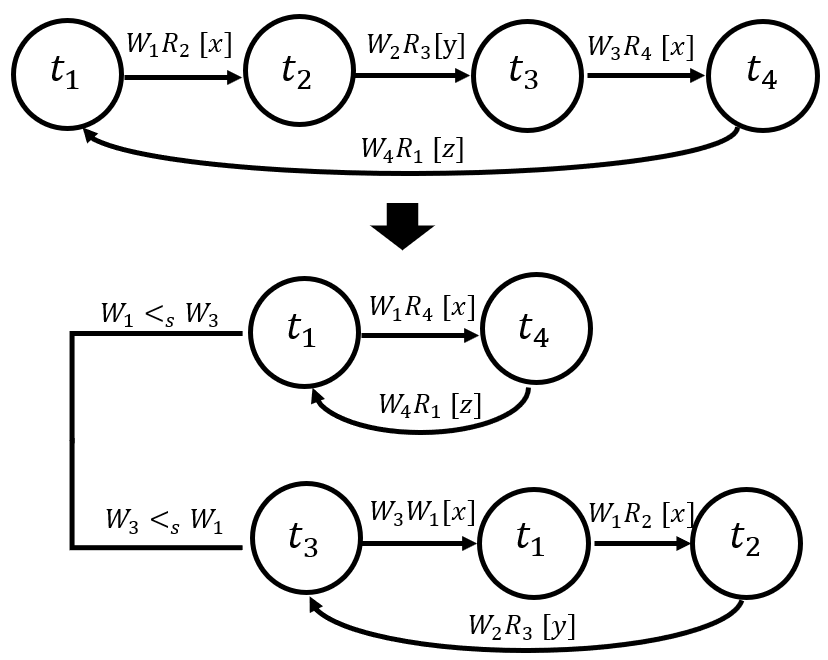}
      \caption{Reduction processes of 5-transactions partially ordered cycle .}
      \label{theorpic4}
     \end{figure}
\end{example}

Based on the theorem \ref{theorem3} and the theorem \ref{theorem4}, we can merge the operation on the same objects and the longest operations are two edges as $RW[x]-WR[x]$ for a conflict cycle graph. Therefore, the edges of the conflict cycle graph can be sorted according to the order for appearance of the objects. And there are only $2N_{obj}$ edge sets at most (figure\ref{figure2}), where $N_{obj} $ is the number of objects in a cycle.

\begin{figure}[h]
  \centering
  \includegraphics[width=0.8 \linewidth]{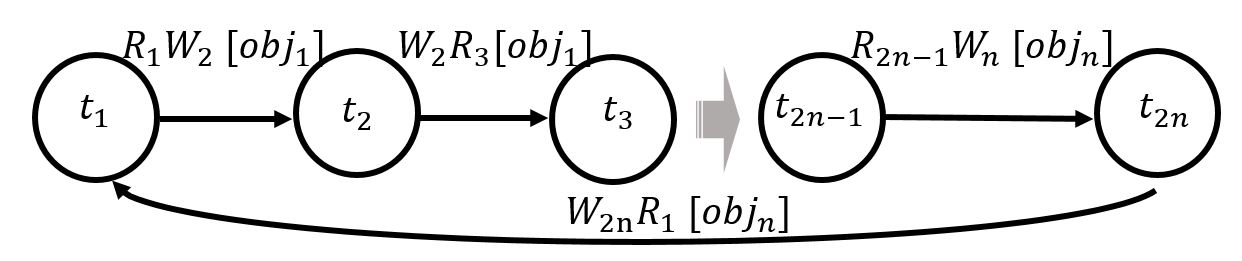}
  \caption{The maximum conflict cycle graph of multi-transactions with $N$-objects.}
  \label{figure2}
\end{figure}

\begin{example}
  We summarize all possible forms of the one-object conflict cycle graph. According to the theorem \ref{theorem1} and \ref{theorem2}, there are 2 transactions in the cycle. The possible combinations are as follows \label{example3}
  $$
  R_1W_2[x]-R_2W_1[x]; \quad R_1W_2[x]-W_2R_1[x]; \quad R_1W_2[x]-W_2W_1[x];
  $$
  $$
  W_1R_2[x]-R_2W_1[x]; \quad W_1R_2[x]-W_2R_1[x]; \quad W_1R_2[x]-W_2W_1[x];
  $$
  $$
  W_1W_2[x]-R_2W_1[x]; \quad W_1W_2[x]-W_2R_1[x]; \quad W_1W_2[x]-W_2W_1[x].
  $$
\end{example}

\begin{example}
  We summarize all possible forms of the two-objects conflict cycle graph. Suppose that the two-objects conflict cycle graph is $G = \{ V, E\}$, where $V$ is the set of transactions, $E$ is conflicts as $\{ WW, WR, RW \}$, and the objects is $\{ x,y \}$. According to the theorem \ref{theorem3} and the theorem \ref{theorem4}, there are at most 4 transactions in the cycle, and the corresponding conflicts is $R_1W_2[x]-W_2R_3[x]-R_3W_4[y]-W_4R_1[y]$.

  If there are three transactions in the cycle, the possible combinations are as follows.
  $$
  R_1W_2[x]-R_2W_3[y]-W_3R_1[y]; \quad R_1W_2[x]-W_2R_3[x]-R_3W_1[y];
  $$
  $$
  W_1R_2[x]-R_2W_3[y]-W_3R_1[y]; \quad R_1W_2[x]-W_2R_3[x]-W_3R_1[y];
  $$
  $$
  W_1W_2[x]-R_2W_3[y]-W_3R_1[y]; \quad R_1W_2[x]-W_2R_3[x]-W_3W_1[y].
  $$
  If there are two transactions in the cycle, the possible combinations are as follows.
  $$
  R_1W_2[x]-R_2W_1[y]; \quad R_1W_2[x]-W_2R_1[y]; \quad R_1W_2[x]-W_2W_1[y];
  $$
  $$
  W_1R_2[x]-R_2W_1[y]; \quad W_1R_2[x]-W_2R_1[y]; \quad W_1R_2[x]-W_2W_1[y];
  $$
  $$
  W_1W_2[x]-R_2W_1[y]; \quad W_1W_2[x]-W_2R_1[y]; \quad W_1W_2[x]-W_2W_1[y].
  $$
\end{example}

We combine conflicting partial orders containing predicate relations constructing a cycle.  If the second type of conflict is contained in the cycle, it is a predicate type data anomaly. Otherwise, it is an entity type data anomaly.  All the combinations of bilateral cycles are shown in Table \ref{table:predicate_anomalies}.

\begin{table*}[ht]
\caption{Data anomaly classification, formal expression, and their edges combinations in the conflict cycles. $N_D$ stands for the number of edges.} 
\footnotesize
\setlength\extrarowheight{1.5pt}

\begin{tabular}{|c|c|l|l|l|}
\hline
\multicolumn{2}{|c|}{\textbf{Types   of Anomalies}} & \multicolumn{1}{c|}{\textbf{Anomalies}}  & \multicolumn{1}{c|}{\textbf{Mathematical   Patterns}}                                                        & \multicolumn{1}{c|}{\textbf{Edges   Combinations}}                                                                               \\ \hline
\multirow{9}{*}{RAT}      & SDA                     & Dirty   Read                             & $W_i   [x_m ] \dots R_j [x_m] \dots A_i$                                                                     & $W_i R_j [x] -R_j A_i   [x] $                                                                                                    \\ \cline{2-5} 
                          & SDA                     & Non-repeatable   Read                    & $R_i   [x_m ] \dots W_j [x_{m+1} ] \dots R_i [x_{m+1} ]$                                                     & $R_i W_j [x] -W_j R_i   [x] $                                                                                                    \\ \cline{2-5} 
                          & SDA                     & Intermediate   Read                      & $W_i   [x_m ] \dots R_j [x_m ] \dots (C_j) \dots W_i [x_{m+1} ]$                                             & $W_i R_j [x] -R_j   W_i/R_j C_j W_i [x] $                                                                                        \\ \cline{2-5} 
                          & \textbf{DDA}            & \textbf{Write-Read   Skew Committed}     & \textbf{$W_i   [x_m ] \dots R_j [x_m ] \dots W_j [y_n ] \dots C_j \dots R_i [y_n ]$}                         & \textbf{$W_i R_j [x]-W_j C_j   R_i [y]$}                                                                                         \\ \cline{2-5} 
                          & \textbf{DDA}            & \textbf{Double-Write   Skew 1 Committed} & \textbf{$W_i   [x_m ] \dots R_j [x_m ] \dots W_j [y_n ] \dots C_j \dots W_i [y_{n+1} ]$}                     & \textbf{$W_i R_j [x]-W_j C_j   W_i [y]$}                                                                                         \\ \cline{2-5} 
                          & \textbf{DDA}            & \textbf{Write-Read   Skew}               & \textbf{$W_i   [x_m ] \dots R_i [x_m ] \dots W_j [y_n ] \dots R_i [y_n ]$}                                   & \textbf{$W_i R_j [x]-W_j R_i   [y]$}                                                                                             \\ \cline{2-5} 
                          & DDA                     & Read   Skew                              & $R_i   [x_m ] \dots W_j [x_{m+1} ] \dots W_j [y_n ] \dots R_i [y_n ]$                                        & $R_i W_j [x]-W_j R_i   [y]$                                                                                                      \\ \cline{2-5} 
                          & \textbf{DDA}            & \textbf{Read   Skew 2}                   & \textbf{$W_i   [x_m ] \dots R_j [x_m ] \dots R_j [y_n ] \dots (C_j) \dots W_i [y_{n+1} ]$}                   & \textbf{$W_i R_j [x]-R_j W_i   [y]  / C_j [y]$}                                                                                  \\ \cline{2-5} 
                          & \textbf{MDA}                     & \textbf{Step   RAT}                               & $   \dots W_{i} [x_m] \dots R_{i} [x_m] \dots $, and $N_D \geq 3$                                    &                                                                                                                                  \\ \hline
\multirow{15}{*}{WAT}     & SDA                     & Dirty   Write                            & $W_i   [x_m ] \dots W_j [x_{m+1} ] \dots A_i/C_i$                                                            & $W_i W_j [x] -W_j   A_i/C_i [x] $                                                                                                \\ \cline{2-5} 
                          & \textbf{SDA}            & \textbf{Lost   Self Update Committed}    & \textbf{$W_i   [x_m ] \dots W_j [x_{m+1} ] \dots C_j \dots R_i [x_{m+1} ]$}                                  & \textbf{$W_i W_j/W_i R_j [x]   -W_j C_j R_i [x]/W_i [x] $}                                                                       \\ \cline{2-5} 
                          & \textbf{SDA}            & \textbf{Full-Write   Committed}          & \textbf{$W_i   [x_m ] \dots W_j [x_{m+1} ] \dots C_j \dots W_i [x_{m+2} ]$}                                  & \textbf{$W_i W_j [x] -R_j C_j   W_i [x] $}                                                                                       \\ \cline{2-5} 
                          & \textbf{SDA}            & \textbf{Full-Write}                      & \textbf{$W_i   [x_m ] \dots W_j [x_{m+1} ] \dots W_i [x_{m+2} ]$}                                            & \textbf{\begin{tabular}[c]{@{}l@{}}$W_i   W_j/W_i R_j  [x] -W_j W_i [x] $ or   \\      $W_i W_j [x] -R_j W_i [x] $\end{tabular}} \\ \cline{2-5} 
                          & SDA                     & Lost   Update                            & $R_i   [x_m ] \dots W_j [x_{m+1} ] \dots W_i [x_{m+2} ]$                                                     & $R_i W_j [x] -W_j   W_i/R_j W_i [x] $                                                                                            \\ \cline{2-5} 
                          & \textbf{SDA}            & \textbf{Lost   Self Update}              & \textbf{$W_i   [x_m ] \dots W_j [x_{m+1} ] \dots R_i [x_{m+1} ]$}                                            & \textbf{$W_i W_j/W_i R_j [x]   -W_j R_i [x] $}                                                                                   \\ \cline{2-5} 
                          & \textbf{DDA}            & \textbf{Double-Write   Skew 2 Committed} & \textbf{$W_i   [x_m ] \dots W_j [x_{m+1} ] \dots W_j [y_n ] \dots C_j \dots R_i [y_n ]$}                     & \textbf{$W_i W_j [x]-W_j C_j   R_i [y]$}                                                                                         \\ \cline{2-5} 
                          & \textbf{DDA}            & \textbf{Full-Write   Skew Committed}     & \textbf{$W_i   [x_m ] \dots W_j [x_{m+1} ] \dots W_j [y_n ] \dots  C_j \dots W_i [y_{n+1} ]$}                & \textbf{$W_i W_j [x]-W_j C_j   W_i [y]$}                                                                                         \\ \cline{2-5} 
                          & \textbf{DDA}            & \textbf{Full-Write   Skew}               & \textbf{$W_i   [x_m ] \dots W_j [x_{m+1} ] \dots W_j [y_n ] \dots W_i [y_{n+1} ]$}                           & \textbf{$W_i W_j [x]-W_j W_i   [y]$}                                                                                             \\ \cline{2-5} 
                          & \textbf{DDA}            & \textbf{Double-Write   Skew 1}           & \textbf{$W_i   [x_m ] \dots R_j [x_m ] \dots W_j [y_n ] \dots W_i [y_{n+1} ]$}                               & \textbf{$W_i R_j [x]-W_j W_i   [y]$}                                                                                             \\ \cline{2-5} 
                          & \textbf{DDA}            & \textbf{Double-Write   Skew 2}           & \textbf{$W_i   [x_m ] \dots W_j [x_m ] \dots W_j [y_n ] \dots W_i [y_{n+1}]$}                                & \textbf{$W_i W_j [x]-W_j R_i   [y]$}                                                                                             \\ \cline{2-5} 
                          & \textbf{DDA}            & \textbf{Read-Write   Skew 1}             & \textbf{$R_i   [x_m ] \dots W_j [x_{m+1} ] \dots W_j [y_n ] \dots W_i [y_{n+1} ]$}                           & \textbf{$R_i W_j [x]-W_j W_i   [y]$}                                                                                             \\ \cline{2-5} 
                          & \textbf{DDA}            & \textbf{Read-Write   Skew 2}             & \textbf{$W_i   [x_m ] \dots W_j [x_{m+1} ] \dots R_j [y_n ]    \dots (C_j) \dots W_i [y_{n+1} ]$}            & \textbf{\begin{tabular}[c]{@{}l@{}}$W_i   W_j [x]-R_j (C_j)W_i [y]$ or\\      $W_i W_j [x]-R_j C_j W_i [y]$\end{tabular}}        \\ \cline{2-5} 
                          & \multirow{2}{*}{\textbf{MDA}}    & \multirow{2}{*}{\textbf{Step   WAT}}              & $   \dots W_{i} [x_m] \dots W_{i} [x_{m+1}] \dots$, and $N_D \geq 3$,                                        &                                                                                                                                  \\ 
                          &                         &                                          & and not include $( \dots W_{j} [x_m ] \dots   R_{j} [x_m] \dots )$                                   &                                                                                                                                  \\ \hline
\multirow{8}{*}{IAT}      & \textbf{SDA}            & \textbf{Non-repeatable   Read Committed} & \textbf{$R_i   [x_m ] \dots W_j [x_{m+1} ] \dots C_j \dots R_i [x_{m+1} ]$}                                  & \textbf{$R_i W_j [x] -W_j C_j   R_i [x] $}                                                                                       \\ \cline{2-5} 
                          & \textbf{SDA}            & \textbf{Lost   Update Committed}         & \textbf{$R_i   [x_m ] \dots W_j [x_{m+1} ] \dots C_j \dots W_i [x_{m+2} ]$}                                  & \textbf{$R_i W_j [x] -W_j C_j   W_i/R_j C_j W_i [x] $}                                                                           \\ \cline{2-5} 
                          & \textbf{DDA}            & \textbf{Read   Skew Committed}           & \textbf{$R_i   [x_m ] \dots W_j [x_{m+1} ] \dots W_j [y_n ] \dots C_j \dots R_i [y_n ]$}                     & \textbf{$R_i W_j [x]-W_j C_j   R_i [y]$}                                                                                         \\ \cline{2-5} 
                          & \textbf{DDA}            & \textbf{Read-Write   Skew 1 Committed}   & \textbf{$R_i   [x_m ] \dots W_j [x_{m+1} ] \dots W_j [y_n ] \dots C_j \dots W_i [y_{n+1} ]$}                 & \textbf{$R_i W_j [x]-W_j C_j   W_i [y]$}                                                                                         \\ \cline{2-5} 
                          & \multirow{2}{*}{DDA}    & \multirow{2}{*}{Write   Skew}            & \multirow{2}{*}{$R_i   [x_m ] \dots W_j [x_{m+1} ] \dots R_j [y_n ]    \dots (C_j)  \dots W_i [y_{n+1}   ]$} & \multirow{2}{*}{$R_i W_j [x]-R_j   (C_j)W_i [y] $}                                                                                                \\ 
                          &                         &                                          &                                                                                                              &                                                                                                                                  \\ \cline{2-5} 
                          & \multirow{2}{*}{\textbf{MDA}}    & \multirow{2}{*}{\textbf{Step   IAT}}              & Not   include $( \dots W_{i} [x_m ] \dots R_{i} [x_m] \dots $                                            &                                                                                                                                  \\ 
                          &                         &                                          & and  $   \dots W_{j}[x_m ] \dots W_{j}[x_{m+1}] \dots )$, $N_D \geq 3$                                   &                                                                                                                                  \\ \hline
\end{tabular}
\label{table:anomaly_classification}
\end{table*}

\subsection{Conflict relations add Status}\label{sec_ConflictRelationsStatus}
 The impact of abort and commit in data consistency is also important operations besides reading and writing. Therefore, we add these two operations to the first type of conflict dependencies. Let $t_i, t_j$ be transactions in a schedule $s$, $t_i <_s t_j$. Two operations $p_i \in t_i$ and $q_j \in t_j $ are conflicts. Let $U$ reflect the status of the transaction still activates. Then conflicts can be extended to several situations below:
  \begin{itemize}
    \item[1.] $p_i-C_i-q_j-A_j/C_j/U_j$: The transaction $t_i$ is committed before the $t_j$ operations;
    \item[2.] $p_i-A_i-q_j-A_j/C_j/U_j$: The transaction $t_i$ is aborted before the $t_j$ operations;
    \item[3.] $p_i-q_j-C_i-A_j/C_j/U_j$: The transaction $t_i$ is committed after the $t_j$ operations;
    \item[4.] $p_i-q_j-A_i-A_j/C_j/U_j$: The transaction $t_i$ is aborted after the $t_j$ operations;
    \item[5.] $p_i-q_j-C_j-A_i/C_i/U_i$: The transaction $t_j$ is committed after the conflict operations;
    \item[6.] $p_i-q_j-A_j-A_i/C_i/U_i$: The transaction $t_j$ is aborted after the conflict operations.
    \item[7.] $p_i-q_j=\{ p_i-U_i-q_j-U_j,p_i-q_j-U_i-U_j, p_i-q_j-U_j-U_i \}$: Both transactions $t_i$ and $t_j$ are still active.
  \end{itemize} 

  Where situations $1-4$ are $t_i$ submitted or rolled back before $t_j$ operation, denoted as $(p_i, A_i/C_i, q_j)$ or $(p_i, q_j, A_i/C_i )$.  Situations $5-6$ are $t_j$ submitted or rolled back before $t_i$ operation, denoted as $(p_i, q_j, A_j/C_j)$; situation $7$ describes that both transactions in scheduling $s$ are not completed, denoted as $(p_i, q_j)$. i.e.,
  $$conf(s)= \{(p_i,A_i/C_i,q_j ), (p_i,q_j,A_i/C_i ), (p_i,q_j,A_j/C_j), (p_i,p_j) \}$$. 
  In the situation $2$, due to the timely rollback of $t_i$, the operation $p_i$ will not affect the operation $q_j$. So do the situation $6$. Substituting the three conflict dependencies $p q\in \{W W,W R,R W\}$ into the remaining 5 cases we can get 15 cases. Among them, $W_i W_j C_j$ indicate the committed $t_j$ after the conflict operations which has the same meaning of $W_i W_j$. So as $W_i R_j$ and $W_i R_j C_j$, and $R_i W_j$ and $R_i W_j C_j$. $W_i R_j C_i$ indicates the $t_i$ committed the operation $W_i$ which is same as $W_i R_j$. $R_i W_j C_i$ and $R_i W_j A_i$ indicates the $t_i$ committed or aborted the operation $R_i$ which cannot disturb the $W_j$ of transaction $t_j$. So the two cases are as the same mean as $R_i W_j$. Through the above discussion, these fifteen cases are sorted into 9 categories that we calls \textbf{partial order pairs(POP)}.
  \begin{itemize}
    \item[1.] $W_i C_i R_j=\{W_iC_iR_j\}$: The version written by $t_i$ and confirmed to be valid is read by $t_j$.
    \item[2.] $W_i C_i W_j =\{ W_i C_i W_j\}$:  The version written by $t_i$ is overwritten by $t_j$ with a newer version, but legal overwriting will not cause inconsistent data status;
    \item[3.] $R_i C_i W_j =\{ R_i C_i W_j\}$: The transaction $t_i$ committed the operation $R_i$ and the $t_j$ rewrite the variable.
    \item[4.] $W_i W_j= \{ W_i W_j, W_i W_j C_j\}$: The version written by $t_i$ is overwritten by $t_j$ with a newer version, so that there may be inconsistent data status;
    \item[5.] $W_i R_j= \{ W_i R_j, W_i R_j C_j, W_i R_j C_i \}$: The version written by $t_i$ read by $t_j$;
    \item[6.] $R_i W_j = \{ R_i W_j, R_i W_j C_j, R_i W_j C_i, R_i W_j A_i\}$: The version read by $t_i$, which is modified by $t_j$ to generate a new version, which may affect $t_i$ to read or modify the same variables;
    \item[7.] $W_i R_j A_i = \{ W_iR_jA_i\}$: The version read by $t_j$ is written by $t_i$. After being rolled back by $t_i$, $t_j$ may read a non-existent version;
    \item[8.] $W_i W_j C_i = \{ W_iW_jC_i\}$: The version written by $t_i$ is overwritten by the updated version of $t_j$, so that the value of the data item written by $t_j$ cannot be read after the occurrence of $C_i$.
    \item[9.] $W_i W_j A_i = \{ W_iW_jA_i\}$: The version written by $t_j$ was overwritten by $t_i$ with an older version due to rollback. 
  \end{itemize}
  For the \textit{POP} $1-6$, we can refine our conflict graph definition. That is to say, based on \cite{2002Concurrency}, the conflict graph we define only extends the traditional conflict relations types.
  
\begin{table}[]
\caption{The comparison between conflict serialization graph and conflict graph of this paper.}
\small
\begin{tabular}{p{1.5cm}|p{3cm}|p{3cm}}
\toprule
\hline
    
 &
 Conflict serialization graph
 &
 Extended conflict graph (We)
  \\ \hline

Purpose & Model conflict serialization &  Model data anomalies   \\ \hline 
\#Edge type & 3 (WW,WR,RW) &  9 (in \S \ref{sec_ConflictRelationsStatus})   \\ \hline 
Data structure & Directed cycle graph &  Directed cycle graph   \\ \hline 

Expressive ness & Limited to describe data anomalies, can not describe Dirty write, Dirty read, and Intermediate Reads & Can express all the data anomalies including new reported and predicate-based anomalies in this paper    \\ \hline 

Added value & Correlate to serializable schedule, but not to specific data anomalies & Correlate directed graph, data anomalies, and consistency together  \\ \hline 
\bottomrule

\end{tabular}
\label{table:conflict_graph_comparison}
\end{table}

  Obviously, the combination of the \textit{POP} $7-9$ itself produces a cycle in conflict graph.
  We will refine these three types in detail.
  
  In category $7$, we can split ${W_i R_j A_i}$ into ${W_i R_j [x]}$ and ${R_j A_i[x]}$, that is, category $7$ and ${R_j A_i[x]}$ are in the same variable The combination of operations under $x$, so the category $6$ can be abbreviated as $R_j A_i$, which can be expressed as figure \ref{RA}.
  \begin{figure}[ht!]
    \begin{subfigure}[t]{0.167\textwidth}
        \centering
        \includegraphics[width= 0.8\linewidth]{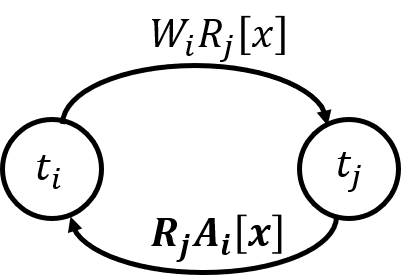}
        \caption{}
        \label{RA}
    \end{subfigure}%
    \begin{subfigure}[t]{0.167\textwidth}
        \centering
        \includegraphics[width= 0.8\linewidth]{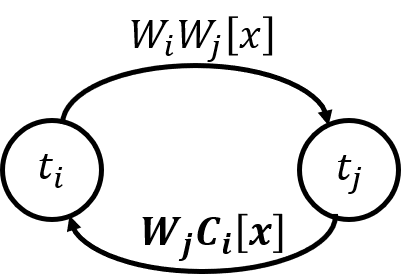}
        \caption{}
        \label{WC}
    \end{subfigure}%
    \begin{subfigure}[t]{0.167\textwidth}
        \centering
        \includegraphics[width= 0.8\linewidth]{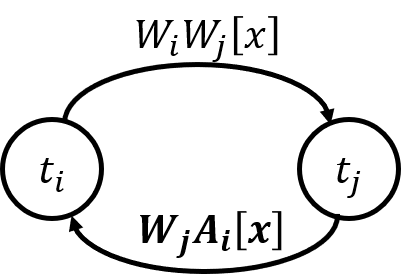}
        \caption{}
        \label{WA}
    \end{subfigure}%
\caption{The partial order (a) $R_j A_i$, (b) $W_j C_i$ or (c) $W_j A_i$ constitutes a cycle.}
  \end{figure}
  In category $8$, we can split ${W_i W_j C_i}$ into ${W_i W_j [x]}$ and ${W_j C_i[x]}$,so the category $8$ can be abbreviated as $W_j C_i$, which can be expressed as figure \ref{WC}.
  
According to Theorems \ref{theorem1} and \ref{theorem2}, a single object can only be a partial order cycle by $2$-transactions. At the same time, the bilateral cycle may be at most two variables. We denote the Single-variable by Double-transactions Anomalies as $SDA$, the Double-variables by Double-transactions Anomalies as $DDA$, and others denoted $MDA$. Therefore, we show all the combined results in Table \ref{table:anomaly_classification} and classify them according to the definition \ref{Types}. Bold fonts in the table indicate 22 newly added anomalies. Therefore, we get all the data anomalies types after classification in Table \ref{table:anomaly_classification}.

In Table \ref{table:anomaly_classification}, we did not distinguish conflicts relations with the predicates domain set. Similarly, if there is a conflicting edge of the predicate-based conflict in the anomaly cycle, it is a predicate-based anomaly. Otherwise, it is an entity-based anomaly. 
The definition of anomalies is also based on the conflict graph(Table \ref{table:conflict_graph_comparison}), \cite{adya2000generalized} only has three kinds of conflict relations(e.g., can not include Dirty Writes, Dirty Reads and Intermediate Reads), we have nine kinds of conflict relations, which makes all data anomalies(e.g., including Dirty Reads) can be included in the extended conflict graph. To the best of our knowledge, this is the first time to completely establish a one-to-one mapping of all data anomalies and conflict graphs.

\begin{definition}(\textbf{Consistency})\label{def:Consistency}
  Consistency = Directed acyclic graph = No data anomalies.
  No Consistency = Directed cyclic graph = Data anomalies.
\end{definition}



\section{Quantitative Research}\label{sec_QuantitativeResearch}
Mastering all data anomalies and being able to classify all data anomalies mean that we have the opportunity to do some quantitative work. Based on data anomalies classification, we quantify the probability of various data anomalies. Then, we quantitatively study the rollback rates of various concurrency control algorithms in different scenarios.

\subsection{Setup}
We evaluate the YCSB \cite{DBLP:conf/cloud/CooperSTRS10} and TPC-C \cite{tpcc} benchmark in a unified transaction processing framework called 3TS\footnote{https://github.com/Tencent/3TS}. 3TS is a open source research framework of concurrent control (CC) protocols based on Deneva \cite{DBLP:journals/pvldb/HardingAPS17}, which carries out quantitative research on data anomalies and CC protocols.

We run all experiments on a single machine with 56 cores (Platinum 8276 CPU@2.20GHz hyper-threading to 112 threads) and 10 TB memory. The statistics of workloads are shown in Table \ref{table:workload_description}. The YCSB contention is controlled by the skew factor of zip distribution access patterns. In our evaluation, low and high contentions are with 0 and 0.9 values for the skew parameter, representing uniform and extremely skew distributions, respectively. The Mix workload consists of 50\% Payment and 50\% NewOrder workloads.

\begin{table}[]
\caption{Workload description. u/r/w/i stands for update rate/read/write/insert, sum of the read, read, and insert is 1. Low contention is with uniform distribution of data access while high contention is with zip distribution and extreme skewness of data access.}
\begin{tabular}{c|c|c|c|c}
\toprule
\hline
 Benchmark&
 No &
  Workload &
  u/r/w/i &
  Skew 
  \\ \hline

\multirow{4}{*}{\textbf{YCSB}} & y1 & Main write & 0.9/0.1/0.9/0 & 0.5  \\ \cline{2-4} 

 & y2& Main read & 0.1/0.9/0.1/0 & 0.5  \\ \cline{2-4} 

 & y3& Low contention & 0.5/0.5/0.5/0 & 0.0  \\ \cline{2-4} 

 & y4& High contention & 0.5/0.5/0.5/0 & 0.9  \\ \hline 

\multirow{4}{*}{\textbf{TPCC}}  & t1& Payment & 1/0.2/0.6/0.2 & - \\ \cline{2-4}  

 & t2& Neworder & 1/0.2/0.6/0.2 & -  \\ \cline{2-4} 

 & t3& Mix & 1/0.2/0.6/0.2 & - \\ \hline 

\bottomrule

\end{tabular}
\label{table:workload_description}
\end{table}

\subsection{Quantitative Data Anomalies}\label{sec_TestandQuantitativeDataAnomalies}
Understanding the statistics of all anomalies explains the current perspective of defining anomalies and isolation levels.
In this part, we analyze the property of syntactic histories and analyze the anomaly statistics from these histories. Internal comprehension of data anomalies is the foundation to discuss isolation levels as well as the CC protocols.  
We later show that the statistics summary is representative of real-life applications, as those anomalies with higher occurrence probabilities in syntactic histories are also well-known discovered ones in real-life.  We first define the test set in the following:

\begin{definition}
\label{def:hts}
\textsc{\textbf{History Testing Set}}

$\mathcal{H}(m,n,k) \overset{\text{def}}{\equiv \joinrel \equiv \joinrel \equiv} \{H\ |\ H.variables = m \wedge H.transactions = n \wedge H.writes + H.reads < k\}$.
\end{definition}
where  $H.variables$, $H.transactions$, $H.writes$, and $H.reads$ denote the number of variables, transactions, write operations, and read operations in history H, respectively. The history testing set $\mathcal{H}(m,n,k)$ is the set of full permutation of all histories that have $m$ variables, $n$ transactions, and $k$ write/read operations. We infer that $\mathcal{H}$ has the following two properties, where the function $Mutex(\mathcal{H}_1, \mathcal{H}_2)$ represents $\mathcal{H}_1$ and $\mathcal{H}_2$ are mutually exclusive.:
\begin{itemize}
    \item If $m_1 \neq m_2$ or $n_1 \neq n_2$ or $k_1 \neq k_2$, then $\mathcal{H}_1(m_1,n_1,k_1)$ and $\mathcal{H}_2(m_2,n_2,k_2)$ are mutually exclusive, i.e., \\$\forall m_1,m_2,n_1,n_2,k_1,k_2 \in \mathbb{N} \quad m_1 \neq m_2 \vee n_1 \neq n_2 \vee k_1 \neq k_2 \Rightarrow Mutex(\mathcal{H}_1(m_1,n_1,k_1),\mathcal{H}_2(m_2,n_2,k_2))$.
    \item For any $m_1,m_2,n_1,n_2$, if $k_1 \leq k_2$, then ${H}_1(m_1,n_1,k_1)$ is subset of ${H}_2(m_2,n_2,k_2)$, i.e., $\forall m_1,m_2,n_1,n_2,k_1,k_2 \in \mathbb{N}\\ k_1 \leq k_2 \Rightarrow \mathcal{H}_1(m_1,n_1,k_1) \subseteq {H}_2(m_2,n_2,k_2)$.
\end{itemize}

For example, $H_1$:$W_1[X_0],R_2[X_0],A_1$ and $H_2$: $W_1[X_0],W_2[X_1],A_1$ are example histories of Dirty Read and Dirty Write anomalies (formally described in Table \ref{table:anomaly_classification}), respectively. They both belong to $\mathcal{H}(1,2,3)$, as they have one variable, two transactions, and three operations involved.

\begin{table}[]
\caption{Evaluation statistic of all data anomalies in static histories generated by by static method (Definition \ref{def:hts} with $\mathcal{H}(3,4,7)$). We compute the percentage of each anomaly as well as each category of RAT, WAT, and IAT. In each category, we rank the anomalies by the highest percentage.}
\small

\begin{tabular}{c|c|c|c}
\toprule
\hline
    Class               & Sub- &               Anomaly name                                    &  Perc                             \\ \hline
\multirow{9}{*}{RAT}  & SDA & Dirty Read                       & 16.68\% \\ 
                   & SDA & Intermediate Read                & 5.86\%                             \\
                   & SDA & Non-repeatable Read             & 5.34\%                             \\
                   & DDA & Read Skew 2                     & 1.86\%                             \\
                   & DDA & Write-Read Skew                 & 1.59\%                             \\
                   & DDA & Read Skew                        & 1.53\%                             \\
                   & MDA & Step RAT                         & 0.41\%                            \\
                   & DDA & Write-Read Skew Committed      & 0.26\%                             \\
                   & DDA & Double-Write Skew 1 Committed & 0.26\%                             \\\hline
\multirow{14}{*}{WAT} & SDA & Dirty Write                      & 36.13\% \\
                   & SDA & Full-Write                       & 11.28\%                            \\
                   & SDA & Lost Update                      & 5.31\%                             \\
                   & DDA & Read-Write Skew 2             & 1.76\%                             \\
                   & DDA & Double-Write Skew 1            & 1.58\%                             \\
                   & DDA & Read-Write Skew 1               & 1.53\%                             \\
                   & DDA & Full-Write Skew                 & 1.51\%                             \\
                   & DDA & Double-Write Skew 2            & 1.51\%                             \\
                   & SDA & Full-Write Committed            & 0.74\%                             \\
                   & MDA & Step WAT                         & 0.35\%                             \\
                   & SDA & Lost Self Update               & 0.23\%                             \\
                   & DDA & Full-Write Skew Committed      & 0.23\%                             \\
                   & DDA & Double-Write Skew 2 Committed & 0.23\%                             \\
                   & SDA & Lost Self Update Committed   & 0.04\%                             \\\hline
\multirow{6}{*}{IAT}  & DDA & Write Skew                       & 1.79\%  \\ 
                   & SDA & Non-repeatable Read Committed  & 0.74\%                            \\
                   & SDA & Lost Update Committed           & 0.74\%                             \\
                   & DDA & Read-Write Skew 1 Committed   & 0.23\%                             \\
                   & DDA & Read Skew Committed             & 0.23\%                             \\
                   & MDA & Step IAT                         & 0.02\%      \\\hline
\bottomrule                      
\end{tabular}
\label{table:anomaly_statistics}
\end{table}

\begin{table}[]
\caption{Evaluation statistic of all data anomalies in static histories generated by by static method (Definition \ref{def:hts} with $\mathcal{H}(3,4,7)$). We summarize the statistics based on SDA, DDA, and MDA.}
\small

\begin{tabular}{c|c|c|c|c}
\toprule
\hline
                   & SDA &  DDA   &  MDA         &     Sum                      \\ \hline
RAT	& 27.88\%	&   5.51\%	& 	0.41\%	& 	33.79\% \\\hline
WAT	& 53.74\%	& 	8.36\%	& 	0.35\%	& 	62.45\% \\\hline
IAT	& 1.48\%	& 	2.25\%	& 	0.02\%	& 	3.76\% \\\hline
Sum	& 83.10\%	& 	16.12\%	& 	0.78\%	& 	 \\\hline

\bottomrule                      
\end{tabular}
\label{table:anomaly_statistics_sub}
\end{table}

\begin{table}[]
\caption{Evaluation statistic of all edges percentage by syntactic histories and cycles with $\mathcal{H}(3,2,6)$.}
\scriptsize

\begin{tabular}{c|c|c|c|c|c|c|c|c}
\toprule
\hline
Edge        & RW	& WR	&   WW	&  WA	&  RA	&   WC	& WCR	&  WCW    \\ \hline
History & 23.14\% & 22.01\% & 21.98\%	& 10.18\%	& 10.18\%	& 10.18\%	& 1.17\%	& 1.16\%		 \\\hline
Cycles & 18.90\% & 25.88\% & 32.45\%	& 6.52\%	& 6.10\%	& 7.23\%	& 1.50\%	& 1.41\%		 \\\hline


\bottomrule                      
\end{tabular}
\label{table:anomaly_statistics_edge}
\end{table}

We first show the percentage of different edges in generated history with parameter $\mathcal{H}(3,2,6)$ in Table \ref{table:anomaly_statistics_edge}. For all histories with or without cycles, we generate them with equal probability. However, histories with cycles have more WW (32.45\%) edge, than WR (25.88\%) and RW (18.9\%) edges, meaning WW is more sensitive in forming cycles.
We then generate histories with parameter $\mathcal{H}(3,4,7)$ and compute the percentage of each anomaly in terms of all cycle histories. 
Table \ref{table:anomaly_statistics} depicts the statistics of these syntactic cycle histories. Since a cycle may contain multiple anomalies, we show the result collected by ranking the priority from RAT and WAT to IAT based on Table \ref{table:anomaly_classification}. We also collected the statistics in different orders, but the result shows minor changes.
Among all anomalies, Dirty Write with 36.13$\%$ occurrences is the most one. Dirty Read with 16.68$\%$ is the most one in RAT. We can tell that those with higher probabilities are very close to known anomalies as they also occur frequently in real-life transaction processing. Previous isolation levels are defined based on these known anomalies, leading to solving practical application scenarios. Yet it is not quantitatively defined, as a newly reported anomaly often can not be categorized.
WAT is more than 60$\%$, as the WW is usually the critical POP in a cycle. We later also discussed that WW conflicts are unavoidable and tricky to most CC protocols. IAT is a rare case, and this is where some commercial databases often sacrifice serializability for performance.

We also summarize $\mathcal{H}(3,4,7)$ histories by SDA, DDA, and MDA in Table \ref{table:anomaly_statistics_sub}. The SDA is the most anomalies in RAT and WAT, and in total, it is more than 83$\%$ of all anomalies. SDA occurs more often as they are composed of fewer operations, thus higher probabilities with a fixed number of read/write operations. Again, the high probability of SDA occurrence is also revealed in real-life scenarios, as they are usually known data anomalies. 
The total probability of forming SDA and DDA data anomalies is as high as 99.22$\%$.
We show that the anomaly cycle of SDA and DDA has only two edges in Table \ref{table:anomaly_classification}. Therefore, it should be low cost to detect these anomalies even by the cycle detection method.
The probability of forming MDA data anomalies is only 0.78$\%$. This part of data anomalies may have complicated and dynamic cycle structures with several edges, which can be extremely expensive for cycle detection.

For RAT and WAT, their probability is as high as 96.34$\%$. These two types of data anomalies must be eliminated in order to ensure data consistency. The probability of IAT is only 3.76$\%$, so solving this part of data anomalies at a high cost will bring a certain degree of performance loss, this kind of data anomalies usually depends on the cycle detection algorithm as \cite{weikum2001transactional, DBLP:conf/icde/Durner019}. Identifying cycles of conflicts is usually prohibitive, as the increasing concurrent transactions may produce exponential overhead \cite{weikum2001transactional}, however, after solving RAT and WAT, we can avoid using cycle detection algorithm by eliminating RW conflict as \cite{DBLP:conf/sigmod/CahillRF08, DBLP:conf/eurosys/YabandehF12}.

\textbf{Lesson learned}: CC should eliminate SDA and DDA type data anomalies in a low-cost way. CC also should eliminate RAT and WAT type data anomalies in a low-cost way. For IAT type data anomalies, CC should trade-off between the cycle detection algorithm and rolling back transactions.


In the following part, we use this syntactic test set as static histories to evaluate the rollback rate of CC protocols.

\subsection{Evaluation of Rollback Rates}\label{sec_EvaluationofRollbackRate}
Rollback rate is one of the most important factors affecting performance. Some concurrency control algorithms focus on reducing the rollback rate to improve performance. However, there is no clear specification on what aspects to mitigate the rollback rate. 
To address this, we evaluate the rollback rate by both static and dynamic methods. The static data is generated by the above syntactic histories while the dynamic method uses the real-life benchmark, i.e., TPC-C and YCSB. Observing these rollback rates helps to understand the behaviors of CC protocols in dealing with anomalies or partial patterns of these anomalies.

\subsubsection{The Definition of Rollback Rate}\label{sec_DefinitionRollbackRateDef}
\
\newline
Before comparing the rollback rate with CC protocols, we first introduce the definition of rollback as well as the True and False rollback rate.

\begin{definition}
\textsc{\textbf{Data Anomaly Rollback}}
Forced to rollback due to a data anomaly in concurrent transactions, i.e., True Rollback (TR).
\end{definition}

\begin{definition}
\textsc{\textbf{Non-Data Anomaly Rollback}}
Forced to rollback in concurrent transactions without a data anomaly. Also called False Rollback (FR).
\end{definition}

\begin{definition}
\textsc{\textbf{Algorithm Rollback}}
Rollback due to CC algorithm's decision (AR). It includes TR and FR.
\end{definition}

\noindent \textsc{\textbf{Calculation of Rollback Rate }} First, we define some terminologies used in formulas below. 
\begin{itemize}
    \item \textbf{Rollback Rate} -- The True Rollback Rate (TRR) is the percentage of histories with cycles. The False Rollback Rate (FRR) is the percentage of histories are not cycles but CC protocols still abort based on their rules. We denote algorithm rollback rate as $R_{alg}$.
    \item \textbf{Number Of Transactions} -- We denote the total number of transactions as $N$, number of transactions rolled back by concurrency control algorithm as $N_{alg}$, number of transactions truly rolled back as $N_{true}$.
\end{itemize}

\vspace{0.15cm}

$$R_{alg}=\cfrac{N_{alg}}{N}=TRR+FRR \eqno{(R1)}$$
$$TRR=\cfrac{N_{true}}{N} \eqno{(R2)}$$
$$FRR=R_{alg}-TRR \eqno{(R3)}$$

\vspace{0.15cm}

\begin{table*}[]
\caption{Evaluation result of static rollback rate. TRR and FRR stand for ture and false rollback rate, respectively.}
\small



\begin{tabular}{c|c|c|c|c|c|c|c}

\toprule
\hline
 &
  \multicolumn{1}{c|}{\textbf{ALL}} &
  \multicolumn{1}{c|}{\textbf{OCC}} &
  \multicolumn{1}{c|}{\textbf{MaaT}} &
  \multicolumn{1}{c|}{\textbf{MVTO}} &
  \multicolumn{1}{c|}{\textbf{TO}} &
  \multicolumn{1}{c|}{\textbf{SSI}} &
  \multicolumn{1}{c}{\textbf{No$\_$wait}} \\ \hline
 &
  TRR &
  FRR &
  FRR &
  FRR &
  FRR &
  FRR &
  FRR \\ \hline
\textbf{$\mathcal{H}$ (2,2,6)} & 69.82\%  & 3.48\% & 3.97\% & 1.53\% & 1.68\% & 1.72\% & 0.90\% \\ \hline
\textbf{$\mathcal{H}$ (3,2,6)} & 72.91\%  & 3.36\% & 3.94\% & 0.75\% & 0.87\% & 1.41\% & 6.88\% \\ \hline
\textbf{$\mathcal{H}$ (2,3,6)} & 52.09\%  & 7.71\% & 8.82\% & 3.29\% & 3.73\% & 3.53\% & 2.66\% \\ \hline
\textbf{$\mathcal{H}$ (2,4,6)} & 42.13\%  & 9.83\% & 11.37\% & 4.16\% & 4.80\% & 4.31\% & 30.19\% \\ \hline
\textbf{$\mathcal{H}$ (3,4,6)} & 46.40\%  & 21.86\% & 25.95\% & 4.36\%  & 6.01\% & 8.66\% & 45.41\% \\ \hline
\textbf{$\mathcal{H}$ (3,4,7)} & 58.36\%  & 20.17\% & 23.22\% & 5.28\%  & 6.45\% & 8.92\% & 38.26\% \\ \hline
\bottomrule
\end{tabular}


\label{table:rollbackrate_static}
\end{table*}

\begin{table}[]
\caption{Rollback rate evaluation result of benchmark 
Workloads. The workloads are described in Table \ref{table:workload_description}
}
\small










\begin{tabular}{c|c|c|c|c|c|c}
\toprule
\hline
  Workload &
  OCC &
  MaaT &
  MVTO &
  TO &
  SSI &
  No$\_$wait \\ \hline

 y1  & 0.09\% & 	0.03\% & 	0.01\% & 	0.01\% & 	0.05\% & 	0.02\% \\ \cline{2-7} 

 y2  & 0.00\% & 	0.00\% & 	0.00\% & 	0.00\% & 	0.00\% & 	0.00\% \\ \cline{2-7} 

 y3  & 0.04\% & 	0.01\% & 	0.00\% & 	0.00\% & 	0.01\% & 	0.01\% \\ \cline{2-7} 

 y4  & 14.86\% & 	68.81\% & 	11.54\% & 	9.56\% & 	47.51\% & 	14.56\% \\ \hline 

 t1  & 11.19\% & 	9.33\% & 	0.05\% & 	0.03\% & 	4.98\% & 	4.41\% \\ \cline{2-7}  

 t2  & 5.10\% & 	4.04\% & 	0.19\% & 	0.18\% & 	5.28\% & 	4.71\% \\ \cline{2-7} 

 t3  & 12.69\% & 	11.89\% & 	1.21\% & 	1.54\% & 	5.07\% & 	12.02\% \\ \hline 

\bottomrule

\end{tabular}
\label{table:rollbackrate_dynamic}
\end{table}

\subsubsection{Evaluation of Rollback Rates}\label{sec_StaticRollbackRate}
\
\newline  
TRR is not an indicator of CC protocols, as these aborts are necessary. Instead, FRR is an indicator to measure the CC protocols. Without a full cycle detection, all CC protocols have a false rollback, where non-anomaly transactions will be aborted. Intuitively, the higher FRR the worse the performance. So reducing the FRR has the potential to improve the performance.


We evaluates several CC protocols
, such as, two-phase locking(2PL) \cite{DBLP:journals/cacm/EswarranGLT76} (includes No\_wait and Wait\_die), Timestamp Ordering (TO) \cite{DBLP:journals/csur/BernsteinG81}, MVCC(Mutil-Version Concurrency Control) with TO (MVTO) \cite{DBLP:journals/tods/BernsteinG83}, Optimistic Concurrency Control (OCC) \cite{DBLP:conf/vldb/KungR79}, MaaT(Multi-access as a Transaction) \cite{DBLP:journals/pvldb/MahmoudANAA14}, Serializable Snapshot Isolation(SSI) and Write-snapshot Isolation(WSI) with 3TS.
Table \ref{table:rollbackrate_static} and Table \ref{table:rollbackrate_dynamic} show the evaluation result.
\begin{itemize}
    \item Table \ref{table:rollbackrate_static} is the evaluation result by a static method with multiple set of static histories (Definition \ref{def:hts}). The result shows the rate of false rollback of CC protocols in scheduling with all permutations and combinations. 
    \item Table \ref{table:rollbackrate_dynamic} is the evaluation result by a dynamic method with dynamic concurrent transactions generated by YCSB and TPC-C workloads based on real-life scenarios. Each schedule is dynamically formed during the execution of concurrent transactions. We calculate the rollback rate of each CC algorithm.
\end{itemize}

\textbf{\textit{Static Rollback Rates}}:
For “ALL” column in Table \ref{table:rollbackrate_static}, the number of data anomalies contained in each group of history is determined, so the TRR is also determined. The difference between different algorithms lies in the different FRRs. 

2PL (No$\_$wait) algorithms have the highest FRR, as both read lock and write lock will cause the concurrent write to abort. For all histories with RW, WR, and WW, No$\_$wait will abort them. 

Both MVTO and TO are in the TO category with similar FRR. MVTO is an improvement based on MVCC technology, which improves the concurrency of WW and RW, and especially eliminates the occurrence of WR by reading specific versions based on its start timestamp. SSI by using MVCC also has a relatively low abort rate, as read does not block write and write does not block read.

Both OCC and MaaT belong to OCC algorithms with similar FRR. Both of them have to detect RW and WW, so they will abort when they meet RW or WW conflict. In contrast, MVTO and TO have some cases to wait WW conflict while SSI will abort two consecutive RW instead of one. 


\textbf{\textit{Dynamic Rollback Rates}}:
From the real run of benchmark data, MVTO and TO yield lower abort rates as shown in Table \ref{table:rollbackrate_dynamic}. The reason is that they can allow some WW conflicts to wait instead of the direct abort. 

OCC and MaaT have higher abort rates in High contention (y4) and TPC-C workloads (t1, t2, and t3).
They suffer higher abort rates by more concurrent transactions and higher concurrent threads in contented benchmark workloads. When mixing with Payment and NewOrder workloads (t3), the contention intensifies in reading/updating the warehouse information, leading to higher abort rates for all CC protocols. 
Essentially, these CC protocols eliminate the conflict edges of cycles, i.e., RW, WW, and WR conflicts. For No\_wait, since reads and writes are mutually blocked, it does not allow all conflicts, resulting in the highest abort rates in both syntactic and benchmark workloads. For TO and MVTO, it partially allows WW by queuing some transactions, making it a very low abort rate in TPC-C workloads, where the main aborts are from WW conflicts. 
Interestingly, \cite{DBLP:conf/icde/Durner019} precisely avoids false rollback by means of the cycle detection. It claims to have competitive throughput compared to the-state-of-art algorithms while having lower abort rates. We will discuss more natural ideas on CC algorithms dealing with these conflicts in \S \ref{sec_algorithm}. 

\textbf{Lesson learned}: Different scenarios have different rollback rates. It is necessary to model the rollback rates for different application scenarios. This is our future work.


Next section, We will later design concise and easy-understanding isolation levels based on our quantitative method in \S \ref{sec_definitionIsolationLevels}.






\section{Applications And Evaluation}\label{sec_Applications}
This section shows how to use \textit{Coo} framework (in \S \ref{sec_dataAnomaliesModel}) to, redefine isolation levels (in \S \ref{sec_definitionIsolationLevels}), 
analyze the implementation process of mainstream algorithm (in \S \ref{sec_algorithm}). 

\begin{table*}[ht]
\caption{Performance evaluation result (Transaction per second) of CC algorithms and their performance improvement (in percentage) comparing to READ COMMITTED(RC) and READ UNCOMMITTED(RU) isolation levels. The experiment is run with 64 server thread. Workload is described in Table \ref{table:workload_description}} \label{table:isolation_level_performance}
\small


\begin{tabular}{c|cccc|ccc|ccc}
\toprule \hline
Workload    & SSI   & NRW    & RC & RU  & WAIT\_DIE & RC & RU & NO\_WAIT & RC & RU         \\\hline
y1 & 445552  & 6.16\%   & 6.63\%  & 16.04\%    & 510021        & 2.28\%        & 13.73\%  & 478832       & 10.92\%      & 11.75\%  \\\hline
y2 & 586457  & 1.88\%   & 7.18\%   & 9.04\%  & 464146        & 17.28\%       & 26.44\%  & 485169       & 4.44\%      & 1.93\%  \\\hline
y3 & 553605  & -2.29\%  & 3.10\%   & 10.25\%   & 543264        & -2.53\%       & 15.35\%  & 482065       & 0.59\%       & -1.16\%  \\\hline
y4 & 140474  & 531.19\% & 327.60\% & 706.78\%  & 113114        & 347.01\%      & 436.43\% & 117279       & 297.40\%     & 337.28\% \\\hline
t1 & 374208  & 0.65\%   & 2.98\%  & 51.78\%   & 406490        & 7.36\%        & 28.57\%  & 408835       & 4.20\%       & 28.90\%  \\\hline
t2 & 254915  & 0.22\%   & 1.53\%   & 32.33\%   & 266444        & 4.17\%        & 17.17\%  & 279314       & 4.08\%       & 13.54\% \\\hline
t3 & 341092  & 0.86\%   & -1.47\%   & 37.75\%   & 151961        & 66.38\%        & 212.88\%  & 233255       & 60.04\%       & 173.88\% \\\hline

\bottomrule
\end{tabular}

\end{table*}

\subsection{Isolation levels}\label{sec_iso}
In addition to data anomalies, isolation levels are another important content in transaction processing technology. What is the relationship between them? 
The previous knowledge system can not sufficiently reflect the relationship between them. Limited data anomalies are used to define the isolation level \cite{16, 10.1145/223784.223785, adya2000generalized}, indicating that there is a relationship between them. However, the isolation levels are not discussed based on all data anomalies, making their relationship unclear.
Therefore, this section discusses the relationship between data anomalies and isolation levels.

\subsubsection{Can isolation levels really improve performance?}\label{sec_DeepCongition}
\
\newline
\cite{DBLP:conf/ds/GrayLPT76} discusses the motivation of the original definitions of isolation levels of ANSI/ISO SQL \cite{16}. Early database systems used two-phase locking(2PL) technology to ensure data consistency. The traditional view is that a weak isolation level helps to improve performance. But is it really true?
 Table \ref{table:isolation_level_performance} shows the performance of SSI \cite{DBLP:conf/sigmod/CahillRF08} and 2PL (No$\_$wait and Wait$\_$die) \cite{DBLP:journals/cacm/EswarranGLT76}, and their performance improvement (in percentage) comparing to Read Committed(RC) and Read Uncommitted(RU) isolation levels. We see that the benefit to sacrifice isolation level from serializable into RC (or NRW in our defined level) is not significant (less than 5$\%$ in most cases) except in some cases with extremely high contention between transactions (more than 100$\%$ for all cases). Though the RU level has notable improvement, it does not suit most practical applications.

\subsubsection{The definition of Isolation Levels}\label{sec_definitionIsolationLevels}
\
\newline
There are two existing methods of isolation levels definition. The first is based on limited data anomalies as \cite{16, 10.1145/223784.223785}, however, this method cannot handle when more data anomalies are reported(Table 
\ref{table:anomaly_classification}). For example, \textit{Write-read Skew}, 
which is a new data anomaly reported by this paper, is not intuitively clear which isolation level can avoid it.
The second is based on conflict graph as \cite{adya2000generalized}, however, \cite{adya2000generalized} is not a pure method, \cite{adya2000generalized} uses conflict graph and some specific data anomalies 
to define isolation levels.
We redefine isolation levels based on the research results of this paper. 
We use a pure conflict graph and do not narrow it to some special cases. 
Instead, we specify all possible primitive data anomalies(Table \ref{table:anomaly_classification}). Therefore, we use a unified method to define isolation levels. We think a good definition of isolation levels has two features as follows.

\begin{itemize}
    \item \emph{Maximize Concurrency}. A good definition of isolation levels helps to maximize concurrent access control by using low-cost heuristic rules.
    \item \emph{Simplicity}. A simple definition of isolation levels help to simplify implementation in engineering systems. 
\end{itemize}

We define two isolation levels for all data anomalies in Table \ref{table:isolation_level} as follows.
\begin{itemize}
    \item The first is the \textbf{No Read and Write Data Anomalies(NRW)} level, which forbid all \textit{RAT} and \textit{WAT} data anomalies. RAT type data anomalies can be avoided by \textit{Read committed heuristic rules}. \textit{RAT} type data anomalies can be avoided by applying write locks, but with lower concurrency. It can also be eliminated by detecting the cycle of \textit{WAT}. The \textit{NRW} isolation level is stronger than \textit{Read Uncommitted}, \textit{Read Committed} and \textit{Repeatable Read} isolation level of the ANSI/ISO SQL \cite{16}. 
    \item The second also the strictest level is the \textbf{No Data Anomalies(NA)} level, which allows no data anomaly. In other words, NA provides what is normally considered as serializability. However, different CC algorithms have different implementation enforcement to guarantee serializability. 
    According to the characteristics of \textit{IAT} data anomalies, such data anomalies can only be detected or avoided by eliminating \textit{RW} and/or \textit{WCW} and/or \textit{WCR} edges. 
    \item From the \textbf{NRW} level to \textbf{NA} level, the degree of concurrency decreases, as the constraints increases.
\end{itemize}

\begin{table}[ht]
\caption{Simplified isolation levels in practice and engineering. NP and P stands for Not Possible and Possible, respectively} 
\label{table:isolation_level}
\small
\begin{tabular}{lclcc}
\toprule
\hline
\multicolumn{2}{c}{Types   of Anomalies} & Anomaly name                                  & NRW                                  & NA                                   \\ \hline
                       & SDA                         & Dirty Read                                              & NP                         & NP                         \\ \cline{2-5} 
                       & SDA                         & Non-repeatable Read                                     & NP                         & NP                         \\ \cline{2-5} 
                       & SDA                         & Phantom                                                 & NP                         & NP                         \\ \cline{2-5} 
        
                       & SDA                         & Intermediate Read                                       & NP                         & NP                         \\ \cline{2-5}
                       & DDA                         & Write-read Skew Committed                               & NP                         & NP                         \\ \cline{2-5} 
                       & DDA                         & Double-write Skew 1 Committed                           & NP                         & NP                         \\ \cline{2-5}
                       & DDA                         & Write-read Skew                                         & NP                         & NP                         \\ \cline{2-5}
                       
                       & DDA                         & Read Skew                                               & NP                         & NP                         \\ \cline{2-5} 
                       & DDA                         & Read Skew 2                                             & NP                         & NP                         \\ \cline{2-5}

\multirow{-10}{*}{RAT} & MDA                         & Step RAT                                                & NP                         & NP                         \\ \hline  
                       & SDA & Dirty Write                     & NP & NP \\ \cline{2-5} 
                       & SDA & Full-Write Committed            & NP & NP \\ \cline{2-5}
                        
                       & SDA & Full-Write                      & NP & NP \\ \cline{2-5}
                       & SDA & Lost Update                     & NP & NP \\ \cline{2-5} 
                       & SDA & Lost Self Update                & NP & NP \\ \cline{2-5}
                       & DDA & Double-write Skew 2   Committed & NP & NP \\ \cline{2-5}
                       & DDA & Full-write Skew   Committed     & NP & NP \\ \cline{2-5}
                       & DDA & Full-write Skew                 & NP & NP \\ \cline{2-5} 
                       & DDA & Double-write Skew 1             & NP & NP \\ \cline{2-5} 
                       & DDA & Double-write Skew 2             & NP & NP \\ \cline{2-5}
                       & DDA & Read-write Skew 1               & NP & NP \\ \cline{2-5} 
                       & DDA & Read-write Skew 2               & NP & NP \\ \cline{2-5} 
                        
\multirow{-13}{*}{WAT} & MDA & Step WAT                        & NP & NP \\  \hline
                       & SDA                         & Non-repeatable Read Committed                           & P                      & NP                         \\ \cline{2-5} 
                       & SDA                         & Lost Update Committed                                   & P                             & NP                         \\ \cline{2-5} 
                       & DDA                         & Read Skew Committed                                     & P                             & NP                         \\ \cline{2-5} 
                       & DDA                         & Read-Write Skew 1 Committed                             & P                             & NP                         \\ \cline{2-5} 
                       & DDA                         & Write Skew                                              & P                             & NP                         \\ \cline{2-5} 
                       & DDA                         & Predicate-based Write Skew                              & P                             & NP                         \\ \cline{2-5} 
\multirow{-7}{*}{IAT}  & MDA                         & Step IAT                                                & P                             & NP                         \\  \hline
                       
                       \bottomrule
\end{tabular}
\end{table}

\subsubsection{Isolation Levels and Data Anomalies}\label{sec_IsolationLevelsandDataAnomalies}
\
\newline
We define the isolation levels based on all data anomalies, which helps to fully reflect the relationship between isolation levels and all data anomalies.
According to Table \ref{table:isolation_level}, isolation levels can also be considered as a classification method of data anomalies, but this classification method is from the perspective of the performance of concurrent control.

\textbf{Lesson learned}: The definition of isolation level is not a mystery. According to the classification of data anomalies, it is a best practice to formulate flexible isolation levels to meet the engineering implementation.

\subsection{Concurrent Control Algorithms}\label{sec_algorithm}

We have divided data anomalies into three types of cycles, called RAT, WAT, and IAT. 
Each cycle contains POPs like WR, WW, and RW. There are two ways to eliminate data anomalies. The first is an algorithm based on a cycle detection algorithm, which is an expensive method. The second is a rule-based algorithm, which is a low-cost method. Most CC algorithms strive to eliminate one or some of the data anomalies to avoid cycle detection. 

Weak isolation levels can eliminate some data anomalies. For any isolation levels defined by different methods, the strongest isolation level must eliminate all data anomalies.


Table \ref{table:cc_algorithm_isolation_level} shows the methods or rules used by CC algorithms to avoid different POPs in different isolation levels. In the following, we discuss more detailed rules on how to avoid different data anomalies.


\begin{table*}[ht]
\caption{Analysis of CC algorithms in dealing with POPs to guarantee different isolation levels. RCR, RWA, and 2RWA, stands for Read Committed Rules, RW POP abort, and 2-continuous RW POPs abort, respectively} \label{table:cc_algorithm_isolation_level}
\small
\begin{tabular}{c|c|c|c|c|c|c}
\toprule
\hline
\multicolumn{2}{c|}{Type of cycle}  & RAT        &  WAT           & \multicolumn{3}{c}{IAT}     \\\hline
\multicolumn{2}{c|}{Target POP} & WR        &  WW           & RW    & RW-WCR & RW-WCW     \\\hline
CC protocol &   Levels &         &             &     &      \\\hline
\multirow{4}{*}{2PL}            &   RU & -         &   Lock       &  -     & -   & -    \\\cline{2-7}
             &   RC & RCR       &   Lock       &  -    & -   & -    \\\cline{2-7}
             &   RR & RCR+2PL   &   Lock   &  Lock+2PL & Lock+2PL   & Lock+2PL  \\\cline{2-7}
             &   S  & RCR+2PL   &   Lock   &  Lock+2PL & Lock+2PL   & Lock+2PL    \\\hline
\multirow{4}{*}{2PL+MVCC}     &   RU & -         &   Lock        & -      & -   & -     \\\cline{2-7}
             &   RC & MVCC      &   Lock        &  -     & -   & -     \\\cline{2-7}
             &   RR & MVCC+TO   &   Lock    & Lock+2PL  & MVCC+2PL   & MVCC+2PL \\\cline{2-7}
             &   S &  MVCC+2PL  &   Lock    & Lock+2PL  & MVCC+2PL   & MVCC+2PL    \\\hline
\multirow{4}{*}{TO}          &   RU & -         & TO            & -    & -   & -     \\\cline{2-7}
           &   RC & RCR       & TO            & -    & -   & -     \\\cline{2-7}
           &   RR & RCR+TO    & TO            & TO    & RCR+TO   & RCR+TO     \\\cline{2-7}
           &   S & RCR+TO     & TO            & TO    & RCR+TO   & RCR+TO     \\\hline
\multirow{4}{*}{TO+MVCC}       &   RU & -         & TO            & -    & -   & -    \\\cline{2-7}
      &   RC & MVCC      & TO            &     & -   & -    \\\cline{2-7}
      &   RR & MVCC+TO   & TO            & TO    & MVCC+TO   & MVCC+TO    \\\cline{2-7}
      &   S & MVCC+TO    & TO       & MVCC+TO    & MVCC+TO   & MVCC+TO    \\\hline
\multirow{4}{*}{OCC+ MVCC}     &   RU & -         & Lock          &   -   & -   & -    \\\cline{2-7}
             &   RC & MVCC      & Lock          &  -   & - & -    \\\cline{2-7}
    &   RR & MVCC+TO   & Lock          &   RWA   & RWA   & RWA    \\\cline{2-7}
    &   S &  MVCC+TO & Lock &   RWA  & RWA   & RWA    \\\hline
\multirow{4}{*}{SSI}          &   RU & -          & Lock          &  -    & -   & -    \\\cline{2-7}
         &   RC & SI         & Lock          &  -    & -   & -    \\\cline{2-7}
         &   RR & SI         & SI       &  -    & SI   & SI    \\\cline{2-7}
         &   S  & SI+2RWA     & SI   &  SI+2RWA   & SI+2RWA   & SI+2RWA    \\\hline
\multirow{4}{*}{WSI}          &   RU & -          & Wait          &    -  & -   & -    \\\cline{2-7}
         &   RC & MVCC       & Wait          &    -  & -   & -    \\\cline{2-7}
         &   RR & MVCC+TO    & Wait          &    RWA  & RWA   & RWA    \\\cline{2-7}
         &   S  & MVCC+TO         & Wait          &    RWA  & RWA   & RWA    \\\hline
\bottomrule
\end{tabular}
\end{table*}

\subsubsection{Avoidance of RAT(WR)}
\
\newline
We first discuss the WR POP. Together with WCR, they have 23.18\% in our syntactic histories as shown in Table \ref{table:anomaly_statistics_edge}. And WR exists in RAT and WAT anomalies while WCR exists in IAT anomalies. Based on our statistic in Table \ref{table:anomaly_statistics}, Dirty Read with 16.68\% is the most anomalies that are composited of the WR. 

For ANSI/ISO SQL, Dirty Read often occurs with RU level enabled, as reading uncommitted active writes is allowed. 
While for the RC level, each read operation must read the committed version. This can be enforced by rules that the only committed data can be read, where WR will be eliminated. 
For RR level, this requirement is stronger than RC such that in RR, each read of a transaction should be consistent(the same). The second read can not be affected by the update or insert operation by other transactions. 

Table \ref{table:cc_algorithm_isolation_level} shows the solutions or rules used by CC algorithms to avoid different POPs in different isolation levels. For example, in MVCC \cite{DBLP:journals/tods/BernsteinG83} or MVCC-based protocol like SSI and WSI, the RC level can be obtained by reading the latest committed version via the versions of MVCC. 
RR can be enforced by a transaction-life-time read lock on a data item or MVCC by snapshot technique controlled by timestamps. Instead of reading the latest committed version in RC where updates may affect two different reads, this time, it should always read a fixed version, for example, the latest committed version that is committed before the reading transaction starts.
In purely TO algorithm \cite{DBLP:journals/csur/BernsteinG81}, RR can be achieved such that the read can not be a success if the data modified timestamp is greater than the start timestamp of the read transaction. In our classification, the avoidance of WR indicates RAT-free, as RATs always hold a WR POP in their cycles. 

For this paper, in engineering implementation, RC and SI can be used as rules to eliminate most data anomalies with NRW isolation level. For example, Read Skew (RW-WR) anomaly includes one WR edge, it can be avoided by the RC rule. But for Read Skew Committed (RW-WCR) anomaly, it has one WCR edge, which can not be avoided by RC rule. But SI rule can avoid this anomaly via reading the specific version instead of the new committed version by the technique of TO.

\textbf{Lesson learned}: The data anomalies of RAT can be eliminated by rule-based methods. There is no need to use the cycle detection algorithm.

\subsubsection{Avoidance of WAT(WW)}
\
\newline
We next discuss the WW POP. Together with WCW, they have 23.14\% in our syntactic histories. And WW exists in WAT anomalies while WCW exists in RAT and IAT anomalies.  Based on our statistic in Table \ref{table:anomaly_statistics}, Dirty Write with 36.13\% is the most anomalies that are composited of the WW.  
The WW is the most annoying one, as we show most anomaly cases are composited of WW in Table \ref{table:anomaly_classification}, and in Table \ref{table:anomaly_classification} we show that from the full permutation of syntactic transaction histories, more than 62$\%$ of the anomalies are types of WATs, which are with WW conflict. Formally speaking, an anomaly is a cycle of several POPs. And WWs are the most critical and often unavoidable ones, as these writes from WW can easily form other POPs with other transactions.

For ANSI/ISO SQL, Dirty Writes often is not allowed even with the lowest RU level enabled.
So eliminating WW conflicts is required from the RU to the highest serializable level. So to avoid the WAT anomalies in different isolation levels, it needs to consider avoiding WR or RW in different levels.

We can summarize the WW conflict prevention techniques by two strategies, NO WAIT and WAIT.  We describe these two strategies in the following. 

\noindent \textbf{NO WAIT} \indent In this case, to avoid the WW conflict to form an anomaly cycle, the CC algorithms abort one of the transactions once two transactions are intended to write on the same data item concurrently. For example, 2PL \cite{DBLP:journals/cacm/EswarranGLT76}, which belongs to a rule-based algorithm, implements this technique by acquiring write lock by one transaction and aborting others who are also acquiring write lock on the same record. Interestingly, other CC algorithms may use implicit write locks. In TO algorithm \cite{DBLP:journals/csur/BernsteinG81}, even though no specific write lock is set, the write operation will check the write timestamp instead. The transactions with write operation abort when they try to write on the records that have been modified after these transactions started.

\noindent \textbf{WAIT} \indent In this case, to avoid the WW conflict constructing an anomaly cycle, the CC algorithms arrange transactions with the concurrent writes on the same item to execute in sequence virtually. Here virtually means that a transaction may be arranged earlier than its real execution. 2PL can also implement a WAIT technique, which is to wait instead of aborting when the other transaction holds the write lock. However, it may introduce deadlock in such cases, where two transactions wait for each other by the write locks of different records(We have proved a deadlock is a specific data anomaly \cite{?}). Then either wait die, which often waits in some pre-arranged orders, or deadlock detection, which detects the above cases, is implemented to guarantee proceeding and serializability. 
In MVTO \cite{DBLP:journals/tods/BernsteinG83}, a write may be inserted as an earlier version if there exists no read between this write to its next committed write. Allowing some writes to be committed improves the concurrency as well as the performance.
Likewise in MaaT \cite{DBLP:journals/pvldb/MahmoudANAA14}, each transaction initially maintains the range from zero to infinite, then dynamically adjust this range by other concurrent transaction. If WW conflicts happen in MaaT, then one transaction range's upper bound should be smaller than the lower bound of the range of another transaction. 

For this paper, the avoidance of both the RAT(WR) and WAT(WW) guarantees the NRW isolation level.

\textbf{Lesson learned}: WAT data anomalies can be eliminated by cycle detection algorithm or rule-based method(locking protocol), which can be divided into WAIT DIE and NO WAIT. However, WAIT DIE may cause deadlock, which needs to be eliminated through the deadlock detection method (cycle detection), in this way, it is still the cycle detection algorithm.


\subsubsection{Avoidance of IAT(RW and/or WCW)}
\
\newline
Once the WW and WR are avoided, the percentage of IAT data anomalies is only 3.76\%, as shown in Table \ref{table:anomaly_statistics}. Table \ref{table:anomaly_statistics_edge} also shows that in forming cycles, RW is the least when compared to WW and WR. That is why most efforts of CC algorithms do not focus on removing RW first.  

For ANSI/ISO SQL and the definition of \cite{16}, only the serializable isolation level can eliminate the data anomalies caused by RW, such as Write Skewd(RW-RW) data anomaly.

In SSI, it allows a single RW to appear, while in WSI, the transaction will abort as long as it meets the RW. For the serializable level, OCC and WSI need to abort all RWs while SSI aborts only when two consecutive RWs constructed. 

\begin{lemma}
The avoidance of WR, WW, and RW conflicts guarantees serializability.
\end{lemma}
We can alternatively prove that an anomaly has at least one of WR, WW, or RW POP in its conflict cycle. It can be easily proved by the contradictory. 
An anomaly is a conflict cycle consisting of POPs. Assume the cycle is without WW, WW, and WR, we can prove that POPs with only WCR, WCW, RCW can not form a cycle. As the "C" already determines the commit order between two transactions. So POPs like $W_0C_0W1$ and $W_1C_1W_0$ can not happen at the same time. Therefore, every two transactions of concurrent transactions have a unique order. So it guarantees serializability with the absence of POPs of WR, WW, and RW. 

For this paper, in order to achieve the NA isolation level, we need to eliminate IAT type data anomalies. But for Non-repeatable Read Committed(RW-WCR) and Lost Updata Committed(RW-WCW) data anomalies, we can not use SSI to eliminate them. Non-repeatable Read Committed(RW-WCR) can be eliminated by SI. Lost Updata Committed(RW-WCW) can be eliminated by detecting whether a Write operation and another write operation occur on the same item at the same time, rather than whether there is a WW edge (WW and WCW POPs are different).

\textbf{Lesson learned}: IAT type data anomalies account for a small proportion. Adopting rule-based methods as much as possible will help to improve the performance of CC algorithms. The avoidance of IAT(RW and/or WCW) guarantees the NA isolation level.

\subsubsection{Discussion of SSI and WSI}
\
\newline
Essentially, all algorithms guarantee serializability in a different way by avoiding conflict cycles. In this part, we will discuss how SSI and WSI algorithms avoid these cycles. 
Firstly in SSI \cite{DBLP:conf/sigmod/CahillRF08}, the algorithm is strongly based on the algorithm of Snapshot Isolation (SI) \cite{DBLP:conf/sigmod/BerensonBGMOO95}, which isolation level is stronger than READ COMMITTED.
In SI, the first-updater-win or first-committer-win strategy is applied for the avoidance of WW conflict to guarantee that only one active write can be committed at a time. This corresponds to the above NO WAIT strategy for WW conflicts. And it guarantees every read is to the committed version by MVCC implementation. SI is not serializable and suffers read skew and write skew anomalies \cite{DBLP:conf/sigmod/BerensonBGMOO95}. 
SSI algorithm \cite{DBLP:conf/sigmod/CahillRF08} is serializable with the following steps on top of SI. (i) The conflict types for the SSI cycle can be only RW and WCR, as WR is avoided by reading the committed version while WW is avoided by aborting one of the conflicting transactions. (ii) As SSI proves each cycle has two two-continuous RWs (special case is read skew anomaly with only two RWs), it avoids a cycle by destroying every two-continuous RWs structure. This for sure may abort some non-cycle transactions with two-continuous RWs, but trading off the better performance with less detection cost.

Interestingly, from our IAT anomalies in Table \ref{table:anomaly_classification}, it exists patterns like RW-WCR (e.g., Non-repeatable Read Committed and Read Skew Committed) and RW-WCW (Lost Update Committed and Read-write Skew 1 Committed), which seems that it does not have WW pattern and-two continuous RWs. In fact in SSI, the second R in RW-WCR did not read the new version since it started earlier than the new committed version. Likewise, for RW-WCW, SSI considers a WCW conflict when the last W is executed, as it writes on the version that already has a new version after it started.


An interesting follow-up work, Write Snapshot Isolation (WSI) \cite{DBLP:conf/eurosys/YabandehF12} is proposed to prevent RW conflicts instead of WW conflicts. They claim that unlike SSI blocks between writes, no interruption occurs by writes between concurrent transactions. It is a bit confusing at first impression whether WSI solves WW conflicts or not. In fact, they do solve WW conflicts implicitly by enforcing the second strategy WAIT between concurrent writes. 
The WSI is serializable by the following steps.
(i) The conflict types for the WSI cycle are with WWs and WRs but not RWs as it only detects RW conflicts. 
(ii) WSI enforce all WWs and WRs become WCWs and WCRs by the centralized commit technique \cite{DBLP:conf/dsn/JunqueiraRY11} where writes operations allow to write into the database only when no committed concurrent transactions have modified these writes. The (ii) condition guarantees writing to the database in sequential.
The transaction aborts when occurring RWs. With only WCWs or WCR, there will be no anomalies. The reason is that these WCW and WCR are not possible to construct conflict cycles, as one of the POPs should be WR, WW, or RW.

\subsection{Relation between Data Anomalies, Isolation Levels, and Concurrency Control algorithms}\label{sec_ConcurrencyControlandDataAnomalies}

The purpose of CC algorithms is to ensure data consistency by eliminating all data anomalies. By giving different levels of isolation, CC algorithms aim to eliminate different edges of POP. Therefore by allowing some certain types of edges or cycles, consistency can also be classified according to different isolation levels \cite{Gray1976Granularity}.
We show the statistics in Table \ref{table:anomaly_statistics} that some edges more frequently appear. In these cases, it is required to eliminate by strict rules at first hand. In contrast, some edges are less anomaly-sensitive, which can be dealt with by a post-processing method. 



\section{Related Work}\label{sec_RelatedWork}

To the best of our knowledge, this is the first study about systematically researching and defining data anomalies from the perspective of all anomalies. In our previous work \cite{haixiang2021systematic}, we systematically studied all entity-based data anomalies. This paper not only includes predicate-based data anomalies, but only makes a quantitative study on data anomalies, isolation levels and CC algorithms.


\textbf{Research history of data anomalies.} Well-known data anomalies, including Dirty Write, Dirty Read, Non-repeatable Read, Phantom, Lost update, Read Skew and Write Skew, etc., proposed in \cite{16,10.1145/223784.223785,839388}. These data anomalies were reported in the 1990s and before.
In recent years, there still exist extensive research works that focus on reporting new data anomalies, including Aborted Reads \cite{xie2015high}, Intermediate Read \cite{839388,xie2015high}, Read-only Transaction Anomaly \cite{10.1145/1031570.1031573,Read_Only_Transactions}, Serial-concurrency-phenomenon and Cross-phenomenon \cite{binnig2014distributed}. \cite{cerone2017algebraic} reports Long Fork, Fractured Reads, and Causality Violation. Even, an unnamed data anomaly was reported in \cite{schenkel2000federated}.
For reference, we make a thorough survey on data anomalies and show them in Table~\ref{tab:datano}. This paper redefines known data anomalies and defines the remaining 20+ new data anomalies in all data anomalies formally, and shows them in Table~\ref{table:anomaly_classification}.

\textbf{Predicate-based data anomalies.} We know Phantom is a predicate-based data anomaly reported by \cite{16,10.1145/223784.223785,839388}, Predicate-Based Write Skew is another predicate-based data anomaly reported by \cite{DBLP:journals/tods/FeketeLOOS05}, \cite{839388} discussed some data anomalies and predicate related topics, but did not give how many kinds of predicate-based data anomalies there are, nor did discuss the relationship between entity-based data anomalies and predicate-based data anomalies. We integrate predicates into POPs. 
We unify the relationship between predicates and entity objects before formally defining data anomalies, so as to unify the expression of predicate-based data anomalies and entity-based data anomalies.

\textbf{Elements of data anomaly.} Whether the number of variables and the number of transactions are related to data anomalies has not been discussed. The early reported data anomalies are single-variable or double-variables data anomalies. For example, Dirty Write, Dirty Read, Non-repeatable Read, Phantom, Lost Update, all are single-variable data anomalies. Only Read Skew and Write Skew are double-variables data anomalies. But all of them involve only two concurrent transactions. We found more new data anomalies from the perspective of the number of variables and report them in this paper, for example, Step RAT, Step WAT and Step IAT involve three variables, and we can extend four variables or more variables data anomalies.
 In the conflict graph composed of data anomalies, we not only classify the data anomalies according to the types of edges in the conflict graph, but also classify the data anomalies carefully according to the number of variables in the conflict graph. Table \ref{table:anomaly_statistics} shows that this classification method is meaningful because the proportion of cycles composed of three or more variables is small, for example, the probability of Step RAT, Step WAT, and Step IAT are less than 1$\%$. 



\textbf{Definition of data anomalies.} Although there are so many data anomalies, but it seems that data anomalies are not explicitly defined. There is no universally accepted definition of data anomalies in academia and industry, and the understanding of data anomalies is only in the way of case-by-case.
What are data anomalies? \cite{cerone2018analysing}`s definition is: \textit{This concurrency-control algorithm allows unserializable behaviours, called anomalies}. To our best knowledge, this is the first and only one to explicitly define what is data anomalies. But it is no systematic research on data anomalies, and it defines data anomalies with CC algorithms but each one CC algorithm is only a method and it is not a root reason. Coo framework explicitly defines what data anomalies is (Definition \ref{DEFanomalies}), studies some characteristics of data anomalies and shows them with some lemmas(\S \ref{sec_Applications}), and classifies data anomalies(\S \ref{sec_ConflictGraphs}). 

\textbf{Research on quantification and classification of data anomalies.} The existing research shows that the work of quantifying and classifying data anomalies is carried out within the known and limited range of data anomalies(only including some anomalies in Table \ref{tab:datano}), using serializable and dependency graph \cite{DBLP:books/aw/BernsteinHG87,10.1145/223784.223785,839388,10.1145/1071610.1071615,Elmasri2006Fundamentals,adya1999weak} technology. \cite{jorwekar2007automating} developed a set of tools and methods to automatically detect data anomalies from applications under snapshot isolation technology. \cite{5767927,10.1145/2391229.2391235,zellag2014consistency,10.1145/2213836.2213920}, use a middleware layer that is embedded between the application and the database, according to the serializable theory and dependency graph technology \cite{DBLP:books/aw/BernsteinHG87,10.1145/223784.223785,839388,10.1145/1071610.1071615,Elmasri2006Fundamentals,adya1999weak}, the method quantifies and classifies the data anomalies. However, this method detects data anomalies based on applications such as TPCC \cite{tpcc}, and finds no new data anomalies. \cite{fekete2009quantifying} quantitatively studied the integrity violation rate of data anomalies in different isolation levels under snapshot isolation technology, but did not propose new data anomaly. \cite{10.1145/2213836.2213920} does similar work with us, but it does not make a quantitative comparison of data anomalies for a variety of concurrent algorithms. To the best of our knowledge, this is the first paper to quantify data anomalies based on all data anomalies.
 
\textbf{Data anomalies and Isolation Levels.} \cite{16,10.1145/223784.223785,839388} try to define isolation levels with known and finite data anomalies. Adya \cite{839388} defines isolation levels with conflict graph and three specific data anomalies, but their methods are not unified. 
We extend the definition of the edge of the conflict graph so that the conflict graph can express all data anomalies, so we can completely and flexibly define the isolation level based on all data anomalies.

\textbf{Concurrent Control algorithms.} There are many concurrent control algorithms, include 2PL \cite{DBLP:journals/cacm/EswarranGLT76}, TO \cite{DBLP:journals/csur/BernsteinG81}, OCC \cite{DBLP:conf/vldb/KungR79,10.5555/1286831.1286844,Bayertime1982Dynamic,boksenbaum1984certification,binnig2014distributed}, MVCC \cite{DBLP:journals/tods/BernsteinG83}.
Locking-based approaches make a particular design on the lock types (e.g., read/write/predicate locks) as well as locking mechanism to disallow single-variable data anomalies \cite{Kapali1976The,Gray1976Granularity,1702620,garciamolina1999review,Elmasri2006Fundamentals}.
MVCC-based approaches disallow a set of data anomalies based on snapshot isolation.
To disallow more data anomalies, \cite{DBLP:conf/sigmod/CahillRF08} proposes SSI and \cite{DBLP:conf/eurosys/YabandehF12} proposes WSI algorithm.
Coo framework is potentially used for detecting and reporting all data anomalies and we design the evaluation model of mainstream concurrent control algorithms. 


\section{Conclusion}

We systematically define all data anomalies. Although the experimental data in this paper is based on entity-based data anomalies, it is helpful to quantify data anomalies, rollback rate and isolation levels in concurrent control technology. This work also contributes to the in-depth analysis of various CC algorithms, which brings opportunities to optimize existing CC algorithms and find new CC algorithms. In the future, we will further carry out in-depth quantitative research to better digitally reveal the internal laws of transactions.

\bibliographystyle{ACM-Reference-Format}
\bibliography{reference}


\begin{thebibliography}{50}


\ifx \showCODEN    \undefined \def \showCODEN     #1{\unskip}     \fi
\ifx \showDOI      \undefined \def \showDOI       #1{#1}\fi
\ifx \showISBNx    \undefined \def \showISBNx     #1{\unskip}     \fi
\ifx \showISBNxiii \undefined \def \showISBNxiii  #1{\unskip}     \fi
\ifx \showISSN     \undefined \def \showISSN      #1{\unskip}     \fi
\ifx \showLCCN     \undefined \def \showLCCN      #1{\unskip}     \fi
\ifx \shownote     \undefined \def \shownote      #1{#1}          \fi
\ifx \showarticletitle \undefined \def \showarticletitle #1{#1}   \fi
\ifx \showURL      \undefined \def \showURL       {\relax}        \fi
\providecommand\bibfield[2]{#2}
\providecommand\bibinfo[2]{#2}
\providecommand\natexlab[1]{#1}
\providecommand\showeprint[2][]{arXiv:#2}

\bibitem[\protect\citeauthoryear{??}{16}{1992}]%
        {16}
 \bibinfo{year}{1992}\natexlab{}.
\newblock \bibinfo{booktitle}{\emph{Database Language – SQL}}.
\newblock \bibinfo{publisher}{American National Standard for Information
  Systems}.
\newblock


\bibitem[\protect\citeauthoryear{{Adya}, {Liskov}, and {O'Neil}}{{Adya}
  et~al\mbox{.}}{2000}]%
        {839388}
\bibfield{author}{\bibinfo{person}{A. {Adya}}, \bibinfo{person}{B. {Liskov}},
  {and} \bibinfo{person}{P. {O'Neil}}.} \bibinfo{year}{2000}\natexlab{}.
\newblock \showarticletitle{Generalized isolation level definitions}. In
  \bibinfo{booktitle}{\emph{Proceedings of 16th International Conference on
  Data Engineering (Cat. No.00CB37073)}}. \bibinfo{pages}{67--78}.
\newblock


\bibitem[\protect\citeauthoryear{Adya, Liskov, and O'Neil}{Adya
  et~al\mbox{.}}{2000}]%
        {adya2000generalized}
\bibfield{author}{\bibinfo{person}{Atul Adya}, \bibinfo{person}{Barbara
  Liskov}, {and} \bibinfo{person}{Patrick O'Neil}.}
  \bibinfo{year}{2000}\natexlab{}.
\newblock \showarticletitle{Generalized isolation level definitions}. In
  \bibinfo{booktitle}{\emph{Proceedings of 16th International Conference on
  Data Engineering (Cat. No. 00CB37073)}}. IEEE, \bibinfo{pages}{67--78}.
\newblock


\bibitem[\protect\citeauthoryear{{Adya} and {Liskov}}{{Adya} and
  {Liskov}}{1999}]%
        {adya1999weak}
\bibfield{author}{\bibinfo{person}{Atul {Adya}} {and}
  \bibinfo{person}{Barbara~H. {Liskov}}.} \bibinfo{year}{1999}\natexlab{}.
\newblock \showarticletitle{Weak Consistency: A Generalized Theory and
  Optimistic Implementations for Distributed Transactions}.
\newblock  (\bibinfo{year}{1999}).
\newblock


\bibitem[\protect\citeauthoryear{Bailis, Fekete, Ghodsi, Hellerstein, and
  Stoica}{Bailis et~al\mbox{.}}{2016}]%
        {10.1145/2909870}
\bibfield{author}{\bibinfo{person}{Peter Bailis}, \bibinfo{person}{Alan
  Fekete}, \bibinfo{person}{Ali Ghodsi}, \bibinfo{person}{Joseph~M.
  Hellerstein}, {and} \bibinfo{person}{Ion Stoica}.}
  \bibinfo{year}{2016}\natexlab{}.
\newblock \showarticletitle{Scalable Atomic Visibility with RAMP Transactions}.
\newblock \bibinfo{journal}{\emph{ACM Trans. Database Syst.}}
  \bibinfo{volume}{41}, \bibinfo{number}{3}, Article \bibinfo{articleno}{15}
  (\bibinfo{date}{July} \bibinfo{year}{2016}), \bibinfo{numpages}{45}~pages.
\newblock
\showISSN{0362-5915}
\urldef\tempurl%
\url{https://doi.org/10.1145/2909870}
\showDOI{\tempurl}


\bibitem[\protect\citeauthoryear{Berenson, Bernstein, Gray, Melton, O’Neil,
  and O’Neil}{Berenson et~al\mbox{.}}{1995b}]%
        {10.1145/223784.223785}
\bibfield{author}{\bibinfo{person}{Hal Berenson}, \bibinfo{person}{Phil
  Bernstein}, \bibinfo{person}{Jim Gray}, \bibinfo{person}{Jim Melton},
  \bibinfo{person}{Elizabeth O’Neil}, {and} \bibinfo{person}{Patrick
  O’Neil}.} \bibinfo{year}{1995}\natexlab{b}.
\newblock \showarticletitle{A Critique of ANSI SQL Isolation Levels}. In
  \bibinfo{booktitle}{\emph{Proceedings of the 1995 ACM SIGMOD International
  Conference on Management of Data}} (San Jose, California, USA)
  \emph{(\bibinfo{series}{SIGMOD ’95})}. \bibinfo{publisher}{Association for
  Computing Machinery}, \bibinfo{address}{New York, NY, USA},
  \bibinfo{pages}{1–10}.
\newblock
\showISBNx{0897917316}
\urldef\tempurl%
\url{https://doi.org/10.1145/223784.223785}
\showDOI{\tempurl}


\bibitem[\protect\citeauthoryear{Berenson, Bernstein, Gray, Melton, O'Neil, and
  O'Neil}{Berenson et~al\mbox{.}}{1995a}]%
        {DBLP:conf/sigmod/BerensonBGMOO95}
\bibfield{author}{\bibinfo{person}{Hal Berenson}, \bibinfo{person}{Philip~A.
  Bernstein}, \bibinfo{person}{Jim Gray}, \bibinfo{person}{Jim Melton},
  \bibinfo{person}{Elizabeth~J. O'Neil}, {and} \bibinfo{person}{Patrick~E.
  O'Neil}.} \bibinfo{year}{1995}\natexlab{a}.
\newblock \showarticletitle{A Critique of {ANSI} {SQL} Isolation Levels}. In
  \bibinfo{booktitle}{\emph{{SIGMOD} Conference}}. \bibinfo{publisher}{{ACM}
  Press}, \bibinfo{pages}{1--10}.
\newblock


\bibitem[\protect\citeauthoryear{Bernstein and Goodman}{Bernstein and
  Goodman}{1981}]%
        {DBLP:journals/csur/BernsteinG81}
\bibfield{author}{\bibinfo{person}{Philip~A. Bernstein} {and}
  \bibinfo{person}{Nathan Goodman}.} \bibinfo{year}{1981}\natexlab{}.
\newblock \showarticletitle{Concurrency Control in Distributed Database
  Systems}.
\newblock \bibinfo{journal}{\emph{{ACM} Comput. Surv.}} \bibinfo{volume}{13},
  \bibinfo{number}{2} (\bibinfo{year}{1981}), \bibinfo{pages}{185--221}.
\newblock


\bibitem[\protect\citeauthoryear{Bernstein and Goodman}{Bernstein and
  Goodman}{1983}]%
        {DBLP:journals/tods/BernsteinG83}
\bibfield{author}{\bibinfo{person}{Philip~A. Bernstein} {and}
  \bibinfo{person}{Nathan Goodman}.} \bibinfo{year}{1983}\natexlab{}.
\newblock \showarticletitle{Multiversion Concurrency Control - Theory and
  Algorithms}.
\newblock \bibinfo{journal}{\emph{{ACM} Trans. Database Syst.}}
  \bibinfo{volume}{8}, \bibinfo{number}{4} (\bibinfo{year}{1983}),
  \bibinfo{pages}{465--483}.
\newblock


\bibitem[\protect\citeauthoryear{Bernstein, Hadzilacos, and Goodman}{Bernstein
  et~al\mbox{.}}{1987}]%
        {DBLP:books/aw/BernsteinHG87}
\bibfield{author}{\bibinfo{person}{Philip~A. Bernstein},
  \bibinfo{person}{Vassos Hadzilacos}, {and} \bibinfo{person}{Nathan Goodman}.}
  \bibinfo{year}{1987}\natexlab{}.
\newblock \bibinfo{booktitle}{\emph{Concurrency Control and Recovery in
  Database Systems}}.
\newblock \bibinfo{publisher}{Addison-Wesley}.
\newblock
\showISBNx{0-201-10715-5}
\urldef\tempurl%
\url{http://research.microsoft.com/en-us/people/philbe/ccontrol.aspx}
\showURL{%
\tempurl}


\bibitem[\protect\citeauthoryear{{Bernstein}, {Shipman}, and
  {Wong}}{{Bernstein} et~al\mbox{.}}{1979}]%
        {1702620}
\bibfield{author}{\bibinfo{person}{P.~A. {Bernstein}}, \bibinfo{person}{D.~W.
  {Shipman}}, {and} \bibinfo{person}{W.~S. {Wong}}.}
  \bibinfo{year}{1979}\natexlab{}.
\newblock \showarticletitle{Formal Aspects of Serializability in Database
  Concurrency Control}.
\newblock \bibinfo{journal}{\emph{IEEE Transactions on Software Engineering}}
  \bibinfo{volume}{SE-5}, \bibinfo{number}{3} (\bibinfo{year}{1979}),
  \bibinfo{pages}{203--216}.
\newblock


\bibitem[\protect\citeauthoryear{Binnig, Hildenbrand, Farber, Kossmann, Lee,
  and May}{Binnig et~al\mbox{.}}{2014}]%
        {binnig2014distributed}
\bibfield{author}{\bibinfo{person}{Carsten Binnig}, \bibinfo{person}{Stefan
  Hildenbrand}, \bibinfo{person}{Franz Farber}, \bibinfo{person}{Donald
  Kossmann}, \bibinfo{person}{Juchang Lee}, {and} \bibinfo{person}{Norman
  May}.} \bibinfo{year}{2014}\natexlab{}.
\newblock \showarticletitle{Distributed snapshot isolation: global transactions
  pay globally, local transactions pay locally}.
\newblock  \bibinfo{volume}{23}, \bibinfo{number}{6} (\bibinfo{year}{2014}),
  \bibinfo{pages}{987--1011}.
\newblock


\bibitem[\protect\citeauthoryear{Boksenbaum, Cart, Ferrie, and Pons}{Boksenbaum
  et~al\mbox{.}}{1984}]%
        {boksenbaum1984certification}
\bibfield{author}{\bibinfo{person}{Claude Boksenbaum}, \bibinfo{person}{Michele
  Cart}, \bibinfo{person}{Jean Ferrie}, {and} \bibinfo{person}{Jeanfrancois
  Pons}.} \bibinfo{year}{1984}\natexlab{}.
\newblock \showarticletitle{Certification by Intervals of Timestamps in
  Distributed Database Systems}.
\newblock  (\bibinfo{year}{1984}), \bibinfo{pages}{377--387}.
\newblock


\bibitem[\protect\citeauthoryear{Cahill, R{\"{o}}hm, and Fekete}{Cahill
  et~al\mbox{.}}{2008}]%
        {DBLP:conf/sigmod/CahillRF08}
\bibfield{author}{\bibinfo{person}{Michael~J. Cahill}, \bibinfo{person}{Uwe
  R{\"{o}}hm}, {and} \bibinfo{person}{Alan~D. Fekete}.}
  \bibinfo{year}{2008}\natexlab{}.
\newblock \showarticletitle{Serializable isolation for snapshot databases}. In
  \bibinfo{booktitle}{\emph{{SIGMOD} Conference}}. \bibinfo{publisher}{{ACM}},
  \bibinfo{pages}{729--738}.
\newblock


\bibitem[\protect\citeauthoryear{Cerone and Gotsman}{Cerone and
  Gotsman}{2018}]%
        {cerone2018analysing}
\bibfield{author}{\bibinfo{person}{Andrea Cerone} {and} \bibinfo{person}{Alexey
  Gotsman}.} \bibinfo{year}{2018}\natexlab{}.
\newblock \showarticletitle{Analysing Snapshot Isolation}.
\newblock \bibinfo{journal}{\emph{J. ACM}} \bibinfo{volume}{65},
  \bibinfo{number}{2} (\bibinfo{year}{2018}), \bibinfo{pages}{11:1--11:41}.
\newblock


\bibitem[\protect\citeauthoryear{Cerone, Gotsman, and Yang}{Cerone
  et~al\mbox{.}}{2017}]%
        {cerone2017algebraic}
\bibfield{author}{\bibinfo{person}{Andrea Cerone}, \bibinfo{person}{Alexey
  Gotsman}, {and} \bibinfo{person}{Hongseok Yang}.}
  \bibinfo{year}{2017}\natexlab{}.
\newblock \showarticletitle{Algebraic Laws for Weak Consistency.}
\newblock  (\bibinfo{year}{2017}), \bibinfo{pages}{26:1--26:18}.
\newblock


\bibitem[\protect\citeauthoryear{Cooper, Silberstein, Tam, Ramakrishnan, and
  Sears}{Cooper et~al\mbox{.}}{2010}]%
        {DBLP:conf/cloud/CooperSTRS10}
\bibfield{author}{\bibinfo{person}{Brian~F. Cooper}, \bibinfo{person}{Adam
  Silberstein}, \bibinfo{person}{Erwin Tam}, \bibinfo{person}{Raghu
  Ramakrishnan}, {and} \bibinfo{person}{Russell Sears}.}
  \bibinfo{year}{2010}\natexlab{}.
\newblock \showarticletitle{Benchmarking cloud serving systems with {YCSB}}. In
  \bibinfo{booktitle}{\emph{SoCC}}. \bibinfo{publisher}{{ACM}},
  \bibinfo{pages}{143--154}.
\newblock


\bibitem[\protect\citeauthoryear{Council}{Council}{2010}]%
        {tpcc}
\bibfield{author}{\bibinfo{person}{Transaction Processing~Performance
  Council}.} \bibinfo{year}{2010}\natexlab{}.
\newblock \showarticletitle{TPC Benchmark C (Revision 5.11)}.
  \bibinfo{publisher}{{TPC}}.
\newblock


\bibitem[\protect\citeauthoryear{Durner and Neumann}{Durner and
  Neumann}{2019}]%
        {DBLP:conf/icde/Durner019}
\bibfield{author}{\bibinfo{person}{Dominik Durner} {and}
  \bibinfo{person}{Thomas Neumann}.} \bibinfo{year}{2019}\natexlab{}.
\newblock \showarticletitle{No False Negatives: Accepting All Useful Schedules
  in a Fast Serializable Many-Core System}. In \bibinfo{booktitle}{\emph{35th
  {IEEE} International Conference on Data Engineering, {ICDE} 2019, Macao,
  China, April 8-11, 2019}}. \bibinfo{publisher}{{IEEE}},
  \bibinfo{pages}{734--745}.
\newblock
\urldef\tempurl%
\url{https://doi.org/10.1109/ICDE.2019.00071}
\showDOI{\tempurl}


\bibitem[\protect\citeauthoryear{Elmasri and Navathe}{Elmasri and
  Navathe}{2006}]%
        {Elmasri2006Fundamentals}
\bibfield{author}{\bibinfo{person}{Ramez Elmasri} {and}
  \bibinfo{person}{Shamkant~B. Navathe}.} \bibinfo{year}{2006}\natexlab{}.
\newblock \showarticletitle{Fundamentals of Database Systems, 5/E}.
\newblock  (\bibinfo{year}{2006}).
\newblock


\bibitem[\protect\citeauthoryear{Epp}{Epp}{2010}]%
        {epp2010discrete}
\bibfield{author}{\bibinfo{person}{Susanna~S Epp}.}
  \bibinfo{year}{2010}\natexlab{}.
\newblock \bibinfo{booktitle}{\emph{Discrete mathematics with applications}}.
\newblock \bibinfo{publisher}{Cengage learning}.
\newblock


\bibitem[\protect\citeauthoryear{Eswaran, Gray, Lorie, and Traiger}{Eswaran
  et~al\mbox{.}}{1976a}]%
        {DBLP:journals/cacm/EswarranGLT76}
\bibfield{author}{\bibinfo{person}{Kapali~P. Eswaran}, \bibinfo{person}{Jim
  Gray}, \bibinfo{person}{Raymond~A. Lorie}, {and} \bibinfo{person}{Irving~L.
  Traiger}.} \bibinfo{year}{1976}\natexlab{a}.
\newblock \showarticletitle{The Notions of Consistency and Predicate Locks in a
  Database System}.
\newblock \bibinfo{journal}{\emph{Commun. {ACM}}} \bibinfo{volume}{19},
  \bibinfo{number}{11} (\bibinfo{year}{1976}), \bibinfo{pages}{624--633}.
\newblock


\bibitem[\protect\citeauthoryear{Eswaran, Gray, Lorie, and Traiger}{Eswaran
  et~al\mbox{.}}{1976b}]%
        {Kapali1976The}
\bibfield{author}{\bibinfo{person}{Kapali~P. Eswaran}, \bibinfo{person}{Jim
  Gray}, \bibinfo{person}{Raymond~A. Lorie}, {and} \bibinfo{person}{Irving~L.
  Traiger}.} \bibinfo{year}{1976}\natexlab{b}.
\newblock \showarticletitle{The Notions of Consistency and Predicate Locks in a
  Database System}.
\newblock \bibinfo{journal}{\emph{Communications of the Acm}}
  \bibinfo{volume}{19}, \bibinfo{number}{11} (\bibinfo{year}{1976}),
  \bibinfo{pages}{624--633}.
\newblock


\bibitem[\protect\citeauthoryear{Fekete, Goldrei, and Asenjo}{Fekete
  et~al\mbox{.}}{2009}]%
        {fekete2009quantifying}
\bibfield{author}{\bibinfo{person}{Alan Fekete}, \bibinfo{person}{Shirley~N
  Goldrei}, {and} \bibinfo{person}{Jorge~Perez Asenjo}.}
  \bibinfo{year}{2009}\natexlab{}.
\newblock \showarticletitle{Quantifying isolation anomalies}.
\newblock  \bibinfo{volume}{2}, \bibinfo{number}{1} (\bibinfo{year}{2009}),
  \bibinfo{pages}{467--478}.
\newblock


\bibitem[\protect\citeauthoryear{Fekete, Liarokapis, O’Neil, O’Neil, and
  Shasha}{Fekete et~al\mbox{.}}{2005b}]%
        {10.1145/1071610.1071615}
\bibfield{author}{\bibinfo{person}{Alan Fekete}, \bibinfo{person}{Dimitrios
  Liarokapis}, \bibinfo{person}{Elizabeth O’Neil}, \bibinfo{person}{Patrick
  O’Neil}, {and} \bibinfo{person}{Dennis Shasha}.}
  \bibinfo{year}{2005}\natexlab{b}.
\newblock \showarticletitle{Making Snapshot Isolation Serializable}.
\newblock \bibinfo{journal}{\emph{ACM Trans. Database Syst.}}
  \bibinfo{volume}{30}, \bibinfo{number}{2} (\bibinfo{date}{June}
  \bibinfo{year}{2005}), \bibinfo{pages}{492–528}.
\newblock
\showISSN{0362-5915}
\urldef\tempurl%
\url{https://doi.org/10.1145/1071610.1071615}
\showDOI{\tempurl}


\bibitem[\protect\citeauthoryear{Fekete, O’Neil, and O’Neil}{Fekete
  et~al\mbox{.}}{2004}]%
        {10.1145/1031570.1031573}
\bibfield{author}{\bibinfo{person}{Alan Fekete}, \bibinfo{person}{Elizabeth
  O’Neil}, {and} \bibinfo{person}{Patrick O’Neil}.}
  \bibinfo{year}{2004}\natexlab{}.
\newblock \showarticletitle{A Read-Only Transaction Anomaly under Snapshot
  Isolation}.
\newblock \bibinfo{journal}{\emph{SIGMOD Rec.}} \bibinfo{volume}{33},
  \bibinfo{number}{3} (\bibinfo{date}{Sept.} \bibinfo{year}{2004}),
  \bibinfo{pages}{12–14}.
\newblock
\showISSN{0163-5808}
\urldef\tempurl%
\url{https://doi.org/10.1145/1031570.1031573}
\showDOI{\tempurl}


\bibitem[\protect\citeauthoryear{Fekete, Liarokapis, O'Neil, O'Neil, and
  Shasha}{Fekete et~al\mbox{.}}{2005a}]%
        {DBLP:journals/tods/FeketeLOOS05}
\bibfield{author}{\bibinfo{person}{Alan~D. Fekete}, \bibinfo{person}{Dimitrios
  Liarokapis}, \bibinfo{person}{Elizabeth~J. O'Neil},
  \bibinfo{person}{Patrick~E. O'Neil}, {and} \bibinfo{person}{Dennis~E.
  Shasha}.} \bibinfo{year}{2005}\natexlab{a}.
\newblock \showarticletitle{Making snapshot isolation serializable}.
\newblock \bibinfo{journal}{\emph{{ACM} Trans. Database Syst.}}
  \bibinfo{volume}{30}, \bibinfo{number}{2} (\bibinfo{year}{2005}),
  \bibinfo{pages}{492--528}.
\newblock
\urldef\tempurl%
\url{https://doi.org/10.1145/1071610.1071615}
\showDOI{\tempurl}


\bibitem[\protect\citeauthoryear{Garciamolina}{Garciamolina}{1999}]%
        {garciamolina1999review}
\bibfield{author}{\bibinfo{person}{Hector Garciamolina}.}
  \bibinfo{year}{1999}\natexlab{}.
\newblock \showarticletitle{Review - The Notions of Consistency and Predicate
  Locks in a Database System.}
\newblock \bibinfo{journal}{\emph{ACM Sigmod Digital Review}}
  \bibinfo{volume}{1} (\bibinfo{year}{1999}).
\newblock


\bibitem[\protect\citeauthoryear{Gray, Lorie, Putzolu, and Traiger}{Gray
  et~al\mbox{.}}{1976a}]%
        {DBLP:conf/ds/GrayLPT76}
\bibfield{author}{\bibinfo{person}{Jim Gray}, \bibinfo{person}{Raymond~A.
  Lorie}, \bibinfo{person}{Gianfranco~R. Putzolu}, {and}
  \bibinfo{person}{Irving~L. Traiger}.} \bibinfo{year}{1976}\natexlab{a}.
\newblock \showarticletitle{Granularity of Locks and Degrees of Consistency in
  a Shared Data Base}. In \bibinfo{booktitle}{\emph{{IFIP} Working Conference
  on Modelling in Data Base Management Systems}}.
  \bibinfo{publisher}{North-Holland}, \bibinfo{pages}{365--394}.
\newblock


\bibitem[\protect\citeauthoryear{Gray, Lorie, Putzolu, and Traiger}{Gray
  et~al\mbox{.}}{1976b}]%
        {Gray1976Granularity}
\bibfield{author}{\bibinfo{person}{Jim Gray}, \bibinfo{person}{Raymond~A.
  Lorie}, \bibinfo{person}{Gianfranco~R. Putzolu}, {and}
  \bibinfo{person}{Irving~L. Traiger}.} \bibinfo{year}{1976}\natexlab{b}.
\newblock \showarticletitle{Granularity of Locks and Degrees of Consistency in
  a Shared Data Base}. In \bibinfo{booktitle}{\emph{Readings in database
  systems (3rd ed.)}}. \bibinfo{pages}{365--394}.
\newblock


\bibitem[\protect\citeauthoryear{Hai-Xiang, Xiao-Yan, Chang, Xiao-Yong, Wei,
  and An-Qun}{Hai-Xiang et~al\mbox{.}}{2021}]%
        {haixiang2021systematic}
\bibfield{author}{\bibinfo{person}{Li Hai-Xiang}, \bibinfo{person}{Li
  Xiao-Yan}, \bibinfo{person}{Liu Chang}, \bibinfo{person}{Du Xiao-Yong},
  \bibinfo{person}{Lu Wei}, {and} \bibinfo{person}{Pan An-Qun}.}
  \bibinfo{year}{2021}\natexlab{}.
\newblock \bibinfo{title}{Systematic definition and classification of data
  anomalies in DBMS (English Version)}.
\newblock
\newblock
\showeprint[arxiv]{2110.14230}~[cs.DB]


\bibitem[\protect\citeauthoryear{Harding, Aken, Pavlo, and Stonebraker}{Harding
  et~al\mbox{.}}{2017}]%
        {DBLP:journals/pvldb/HardingAPS17}
\bibfield{author}{\bibinfo{person}{Rachael Harding}, \bibinfo{person}{Dana~Van
  Aken}, \bibinfo{person}{Andrew Pavlo}, {and} \bibinfo{person}{Michael
  Stonebraker}.} \bibinfo{year}{2017}\natexlab{}.
\newblock \showarticletitle{An Evaluation of Distributed Concurrency Control}.
\newblock \bibinfo{journal}{\emph{Proc. {VLDB} Endow.}} \bibinfo{volume}{10},
  \bibinfo{number}{5} (\bibinfo{year}{2017}), \bibinfo{pages}{553--564}.
\newblock


\bibitem[\protect\citeauthoryear{Janin and Walukiewicz}{Janin and
  Walukiewicz}{1996}]%
        {janin1996expressive}
\bibfield{author}{\bibinfo{person}{David Janin} {and} \bibinfo{person}{Igor
  Walukiewicz}.} \bibinfo{year}{1996}\natexlab{}.
\newblock \showarticletitle{On the expressive completeness of the propositional
  mu-calculus with respect to monadic second order logic}. In
  \bibinfo{booktitle}{\emph{International Conference on Concurrency Theory}}.
  Springer, \bibinfo{pages}{263--277}.
\newblock


\bibitem[\protect\citeauthoryear{Jorwekar, Fekete, Ramamritham, and
  Sudarshan}{Jorwekar et~al\mbox{.}}{2007}]%
        {jorwekar2007automating}
\bibfield{author}{\bibinfo{person}{Sudhir Jorwekar}, \bibinfo{person}{Alan
  Fekete}, \bibinfo{person}{Krithi Ramamritham}, {and} \bibinfo{person}{S
  Sudarshan}.} \bibinfo{year}{2007}\natexlab{}.
\newblock \showarticletitle{Automating the detection of snapshot isolation
  anomalies}.
\newblock  (\bibinfo{year}{2007}), \bibinfo{pages}{1263--1274}.
\newblock


\bibitem[\protect\citeauthoryear{Junqueira, Reed, and Yabandeh}{Junqueira
  et~al\mbox{.}}{2011}]%
        {DBLP:conf/dsn/JunqueiraRY11}
\bibfield{author}{\bibinfo{person}{Flavio Junqueira}, \bibinfo{person}{Benjamin
  Reed}, {and} \bibinfo{person}{Maysam Yabandeh}.}
  \bibinfo{year}{2011}\natexlab{}.
\newblock \showarticletitle{Lock-free transactional support for large-scale
  storage systems}. In \bibinfo{booktitle}{\emph{{DSN} Workshops}}.
  \bibinfo{publisher}{{IEEE} Computer Society}, \bibinfo{pages}{176--181}.
\newblock


\bibitem[\protect\citeauthoryear{Kung and Robinson}{Kung and Robinson}{1979}]%
        {DBLP:conf/vldb/KungR79}
\bibfield{author}{\bibinfo{person}{H.~T. Kung} {and} \bibinfo{person}{John~T.
  Robinson}.} \bibinfo{year}{1979}\natexlab{}.
\newblock \showarticletitle{On Optimistic Methods for Concurrency Control}. In
  \bibinfo{booktitle}{\emph{{VLDB}}}. \bibinfo{publisher}{{IEEE} Computer
  Society}, \bibinfo{pages}{351}.
\newblock


\bibitem[\protect\citeauthoryear{Mahmoud, Arora, Nawab, Agrawal, and
  Abbadi}{Mahmoud et~al\mbox{.}}{2014}]%
        {DBLP:journals/pvldb/MahmoudANAA14}
\bibfield{author}{\bibinfo{person}{Hatem~A. Mahmoud}, \bibinfo{person}{Vaibhav
  Arora}, \bibinfo{person}{Faisal Nawab}, \bibinfo{person}{Divyakant Agrawal},
  {and} \bibinfo{person}{Amr~El Abbadi}.} \bibinfo{year}{2014}\natexlab{}.
\newblock \showarticletitle{MaaT: Effective and scalable coordination of
  distributed transactions in the cloud}.
\newblock \bibinfo{journal}{\emph{Proc. {VLDB} Endow.}} \bibinfo{volume}{7},
  \bibinfo{number}{5} (\bibinfo{year}{2014}), \bibinfo{pages}{329--340}.
\newblock


\bibitem[\protect\citeauthoryear{Schenkel, Weikum, Weissenberg, and
  Wu}{Schenkel et~al\mbox{.}}{2000}]%
        {schenkel2000federated}
\bibfield{author}{\bibinfo{person}{Ralf Schenkel}, \bibinfo{person}{Gerhard
  Weikum}, \bibinfo{person}{N Weissenberg}, {and} \bibinfo{person}{Xuequn Wu}.}
  \bibinfo{year}{2000}\natexlab{}.
\newblock \showarticletitle{Federated transaction management with snapshot
  isolation}.
\newblock \bibinfo{journal}{\emph{Lecture Notes in Computer Science}}
  (\bibinfo{year}{2000}), \bibinfo{pages}{1--25}.
\newblock


\bibitem[\protect\citeauthoryear{Schlageter}{Schlageter}{1981}]%
        {10.5555/1286831.1286844}
\bibfield{author}{\bibinfo{person}{Gunter Schlageter}.}
  \bibinfo{year}{1981}\natexlab{}.
\newblock \showarticletitle{Optimistic Methods for Concurrency Control in
  Distributed Database Systems}. In \bibinfo{booktitle}{\emph{Proceedings of
  the Seventh International Conference on Very Large Data Bases - Volume 7}}
  (Cannes, France) \emph{(\bibinfo{series}{VLDB ’81})}.
  \bibinfo{publisher}{VLDB Endowment}, \bibinfo{pages}{125–130}.
\newblock


\bibitem[\protect\citeauthoryear{Smullyan}{Smullyan}{1995}]%
        {smullyan1995first}
\bibfield{author}{\bibinfo{person}{Raymond~M Smullyan}.}
  \bibinfo{year}{1995}\natexlab{}.
\newblock \bibinfo{booktitle}{\emph{First-order logic}}.
\newblock \bibinfo{publisher}{Courier Corporation}.
\newblock


\bibitem[\protect\citeauthoryear{Weikum and Vossen}{Weikum and Vossen}{2001}]%
        {weikum2001transactional}
\bibfield{author}{\bibinfo{person}{Gerhard Weikum} {and}
  \bibinfo{person}{Gottfried Vossen}.} \bibinfo{year}{2001}\natexlab{}.
\newblock \bibinfo{booktitle}{\emph{Transactional information systems: theory,
  algorithms, and the practice of concurrency control and recovery}}.
\newblock \bibinfo{publisher}{Elsevier}.
\newblock


\bibitem[\protect\citeauthoryear{Weikum and Vossen}{Weikum and Vossen}{2002}]%
        {2002Concurrency}
\bibfield{author}{\bibinfo{person}{G. Weikum} {and} \bibinfo{person}{G.
  Vossen}.} \bibinfo{year}{2002}\natexlab{}.
\newblock \showarticletitle{Concurrency Control: Notions of Correctness for the
  Page Model}.
\newblock \bibinfo{journal}{\emph{Transactional Information Systems}}
  (\bibinfo{year}{2002}), \bibinfo{pages}{61--123}.
\newblock


\bibitem[\protect\citeauthoryear{wikipedia}{wikipedia}{[n.d.]}]%
        {Read_Only_Transactions}
\bibfield{author}{\bibinfo{person}{wikipedia}.}
  \bibinfo{year}{[n.d.]}\natexlab{}.
\newblock \bibinfo{title}{Read$\_$Only$\_$Transactions}.
\newblock \bibinfo{howpublished}{Website}.
\newblock
\urldef\tempurl%
\url{https://wiki.postgresql.org/wiki/SSI#Read_Only_Transactions}
\showURL{%
\tempurl}


\bibitem[\protect\citeauthoryear{Xiaoyong}{Xiaoyong}{2017}]%
        {duxiaoyong_read_partial_committed}
\bibfield{author}{\bibinfo{person}{Du~el~at. Xiaoyong}.}
  \bibinfo{year}{2017}\natexlab{}.
\newblock \showarticletitle{Big data management}.
\newblock  (\bibinfo{year}{2017}).
\newblock


\bibitem[\protect\citeauthoryear{{Xie}, {Su}, {Littley}, {Alvisi}, {Kapritsos},
  and {Wang}}{{Xie} et~al\mbox{.}}{2015}]%
        {xie2015high}
\bibfield{author}{\bibinfo{person}{Chao {Xie}}, \bibinfo{person}{Chunzhi {Su}},
  \bibinfo{person}{Cody {Littley}}, \bibinfo{person}{Lorenzo {Alvisi}},
  \bibinfo{person}{Manos {Kapritsos}}, {and} \bibinfo{person}{Yang {Wang}}.}
  \bibinfo{year}{2015}\natexlab{}.
\newblock \showarticletitle{High-performance ACID via modular concurrency
  control}. In \bibinfo{booktitle}{\emph{Proceedings of the 25th Symposium on
  Operating Systems Principles}}. \bibinfo{pages}{279--294}.
\newblock


\bibitem[\protect\citeauthoryear{Yabandeh and Ferro}{Yabandeh and
  Ferro}{2012}]%
        {DBLP:conf/eurosys/YabandehF12}
\bibfield{author}{\bibinfo{person}{Maysam Yabandeh} {and}
  \bibinfo{person}{Daniel~G{\'{o}}mez Ferro}.} \bibinfo{year}{2012}\natexlab{}.
\newblock \showarticletitle{A critique of snapshot isolation}. In
  \bibinfo{booktitle}{\emph{EuroSys}}. \bibinfo{publisher}{{ACM}},
  \bibinfo{pages}{155--168}.
\newblock


\bibitem[\protect\citeauthoryear{{Zellag} and {Kemme}}{{Zellag} and
  {Kemme}}{2011}]%
        {5767927}
\bibfield{author}{\bibinfo{person}{K. {Zellag}} {and} \bibinfo{person}{B.
  {Kemme}}.} \bibinfo{year}{2011}\natexlab{}.
\newblock \showarticletitle{Real-time quantification and classification of
  consistency anomalies in multi-tier architectures}. In
  \bibinfo{booktitle}{\emph{2011 IEEE 27th International Conference on Data
  Engineering}}. \bibinfo{pages}{613--624}.
\newblock


\bibitem[\protect\citeauthoryear{Zellag and Kemme}{Zellag and Kemme}{2012a}]%
        {10.1145/2213836.2213920}
\bibfield{author}{\bibinfo{person}{Kamal Zellag} {and} \bibinfo{person}{Bettina
  Kemme}.} \bibinfo{year}{2012}\natexlab{a}.
\newblock \showarticletitle{ConsAD: A Real-Time Consistency Anomalies
  Detector}. In \bibinfo{booktitle}{\emph{Proceedings of the 2012 ACM SIGMOD
  International Conference on Management of Data}} (Scottsdale, Arizona, USA)
  \emph{(\bibinfo{series}{SIGMOD ’12})}. \bibinfo{publisher}{Association for
  Computing Machinery}, \bibinfo{address}{New York, NY, USA},
  \bibinfo{pages}{641–644}.
\newblock
\showISBNx{9781450312479}
\urldef\tempurl%
\url{https://doi.org/10.1145/2213836.2213920}
\showDOI{\tempurl}


\bibitem[\protect\citeauthoryear{Zellag and Kemme}{Zellag and Kemme}{2012b}]%
        {10.1145/2391229.2391235}
\bibfield{author}{\bibinfo{person}{Kamal Zellag} {and} \bibinfo{person}{Bettina
  Kemme}.} \bibinfo{year}{2012}\natexlab{b}.
\newblock \showarticletitle{How Consistent is Your Cloud Application?}. In
  \bibinfo{booktitle}{\emph{Proceedings of the Third ACM Symposium on Cloud
  Computing}} (San Jose, California) \emph{(\bibinfo{series}{SoCC ’12})}.
  \bibinfo{publisher}{Association for Computing Machinery},
  \bibinfo{address}{New York, NY, USA}, Article \bibinfo{articleno}{6},
  \bibinfo{numpages}{14}~pages.
\newblock
\showISBNx{9781450317610}
\urldef\tempurl%
\url{https://doi.org/10.1145/2391229.2391235}
\showDOI{\tempurl}


\bibitem[\protect\citeauthoryear{Zellag and Kemme}{Zellag and Kemme}{2014}]%
        {zellag2014consistency}
\bibfield{author}{\bibinfo{person}{Kamal Zellag} {and} \bibinfo{person}{Bettina
  Kemme}.} \bibinfo{year}{2014}\natexlab{}.
\newblock \showarticletitle{Consistency anomalies in multi-tier architectures:
  automatic detection and prevention}.
\newblock  \bibinfo{volume}{23}, \bibinfo{number}{1} (\bibinfo{year}{2014}),
  \bibinfo{pages}{147--172}.
\newblock


\end{thebibliography}



\end{document}